\documentclass[3p,final,number]{elsarticle}
\usepackage{mathpazo,charter,amssymb,paralist}
\usepackage[raiselinks=false,colorlinks=true,citecolor=blue,urlcolor=blue,linkcolor=blue,bookmarksopen=true]{hyperref}
\usepackage{xcolor,bookmark}
\usepackage[all,dvips,color,arrow,curve,poly]{xy}
\newtheorem{theorem}{Theorem}

\newtheorem{lemma}[theorem]{Lemma} 
\newenvironment{proof}{\noindent{\bf Proof.}\ }{\hfill\qed\par\bigskip}
\makeatletter\newcommand\setreflabel[1]{\protected@edef\@currentlabel{#1}}\newcounter{claim}
\newenvironment{claim}{\par\refstepcounter{claim}\setreflabel{$\langle$\theclaim$\rangle$}\label{@clm:\theclaim}\vspace{1ex}\noindent\@currentlabel\it}{\par\vspace{1ex}}
\newenvironment{claimplus}{\par\setreflabel{$\langle$\theclaim$^+\rangle$}\vspace{1ex}\noindent\@currentlabel\it}{\par\vspace{1ex}}
\newenvironment{claimaux}[1]{\par\setreflabel{$\langle$#1$\rangle$}\vspace{1ex}\noindent\@currentlabel\it}{\par\vspace{1ex}}
\newenvironment{proofclaim}[1]{\par\noindent#1}{This proves~\ref{@clm:\theclaim}.\par\vspace{1ex}}
\def\ps@pprintTitle{\let\@oddhead\@empty\let\@evenhead\@empty\def\@oddfoot{\footnotesize\itshape\hfill\today}\let\@evenfoot\@oddfoot}
\makeatother\hyphenpenalty=5000\tolerance=1000
\begin{document}
\title{Obstructions to chordal circular-arc graphs of small independence number}
\author[mathew]{Mathew Francis\fnref{mathew-thanks}}\ead{mathew@imsc.res.in}
\author[pavol]{Pavol Hell\fnref{pavol-thanks}}\ead{pavol@sfu.ca}
\author[juraj]{Juraj Stacho\fnref{juraj-thanks}\corref{cor}}\ead{j.stacho@warwick.ac.uk}
\address[mathew]{Institute of Mathematical Sciences, IV Cross Road, CIT Campus, Taramani, Chennai 600 113, India}
\address[pavol]{School of Computing Science, Simon Fraser University, 8888 University Drive, Burnaby, Canada V5A 1S6}
\address[juraj]{DIMAP and Mathematics Institute, University of Warwick, Coventry CV4 7AL, United Kingdom}
\cortext[cor]{Corresponding author}
\fntext[mathew-thanks]{MF partially supported by the grant ANR-09-JCJC-0041.}
\fntext[pavol-thanks]{PH partially supported by the author's NSERC Discovery Grant.}
\fntext[juraj-thanks]{JS gratefully acknowledges support from EPSRC, grant EP/I01795X/1.}


\begin{abstract}
A {\em blocking quadruple} (BQ) is a quadruple of vertices of a graph such that
any two vertices of the quadruple either miss (have no neighbours on) some path
connecting the remaining two vertices of the quadruple, or are connected by some
path missed by the remaining two vertices. This is akin to the notion of
asteroidal triple used in the classical characterization of interval graphs by
Lekkerkerker and Boland \cite{lekker62}.

We show that a circular-arc graph cannot have a blocking quadruple.  We also
observe that the absence of blocking quadruples is not in general sufficient to
guarantee that a graph is a circular-arc graph.  Nonetheless, it can be shown to
be sufficient for some special classes of graphs, such as those investigated in
\cite{bon09}.

In this note, we focus on chordal graphs, and study the relationship between the
structure of chordal graphs and the presence/absence of blocking quadruples.

Our contribution is two-fold.  Firstly, we provide a forbidden induced subgraph
characterization of chordal graphs without blocking quadruples. In particular,
we observe that all the forbidden subgraphs are variants of
the subgraphs forbidden for interval graphs \cite{lekker62}. Secondly, we show
that the absence of blocking quadruples is sufficient to guarantee that a
chordal graph with no independent set of size five is a circular-arc graph. In
our proof we use a novel geometric approach, constructing a circular-arc
representation by traversing around a carefully chosen clique tree.
\end{abstract} 

\maketitle

\section{Introduction}

The study of graph obstructions has a long tradition in graph theory.  To
understand the structure of graphs in a particular graph class, it is often
useful (if not easier) instead to characterize all minimal graphs that are not
in the class, usually known as {\em obstructions}. They often result in
elegant characterization theorems and can be used as succinct certificates
in certifying algorithms.

In this paper, we seek obstructions to {\em circular-arc graphs}, the
intersection graphs of families of arcs of a circle. This problem dates back at
least as far as the 1970's \cite{klee69,trotmoor76,tuck69,tuck70,tuck74}, and remains a
challenging question capturing the interest of many researchers over the years
\cite{bang94,bon09,fhh99,klee69,lin06,trotmoor76,tuck74,tuck69,tuck70}.

Predating the study of circular arc graphs, the class of {\em interval graphs},
intersection graphs of families of intervals of the real line, was investigated.
Interval graphs are a subclass of {\em chordal graphs}, graphs in which every
cycle has a chord, as well as of circular-arc graphs. They are known to admit a
number of interesting characterizations \cite{gilmore64,lekker62} and efficient
recognition algorithms \cite{pq76,derek09,intcert}.  In particular, the result
of Lekkerkerker and Boland \cite{lekker62} describes interval graphs in terms of
forbidden induced subgraphs as well as forbidden substructures -- chordless
cycles and so-called {\em asteroidal triples}.

This result is the main motivation of our paper wherein we seek to describe
analogous forbidden substructures for circular-arc graphs.

We remark in passing that, besides interval graphs, there are other subcases of
circular-arc graphs that have already been characterized by the absence of simple
obstructions.  Namely, {\em unit circular-arc graphs} and {\em proper
circular-arc graphs} in \cite{tuck74}, {\em chordal proper circular-arc
graphs} in \cite{bang94}, {\em cobipartite circular-arc graphs} in
\cite{trotmoor76} and later in \cite{fhh99} (using so-called {\em
edge-asteroids}), and {\em Helly circular-arc graphs} within circular-arc graphs
in \cite{lin06} (using so-called {\em obstacles}). 

More recently, in \cite{bon09}, the authors gave forbidden induced subgraph
characterizations for $P_4$-free circular-arc graphs, diamond-free circular-arc
graphs, paw-free circular-arc graphs, and most relevant for this paper, they
characterized claw-free chordal circular-arc graphs.  Our results (namely
Theorem~\ref{thm:main}) may be seen as complementing their work, since in this
regard we give a forbidden induced subgraph characterization of
$\overline{K_5}$-free chordal circular-arc~graphs.

\section{Blocking quadruple}

To build intuition, we start by recalling the definition of asteroidal triple. 
We say that a vertex $x$ {\em misses} a path $P$ in $G$ if $x$
has no neighbour on $P$. 

Vertices $x,y,z$ form an {\em asteroidal triple} of $G$ if between any two of
them, there is a path in $G$ missed by the third vertex.  It is easy to see that
an interval graph cannot have an asteroidal triple \cite{lekker62}.

We say that vertices $x,y$ {\em avoid} vertices $z,w$ in $G$ if there exists an
$xy$-path missed by both $z$ and $w$, or there exists a $zw$-path missed by both
$x$ and $y$. 

We say that vertices $x,y,z,w$ form a {\em blocking quadruple} (BQ) of $G$ if
any two of them avoid the remaining two. Namely, if $x,y$ avoid $z,w$, if $x,z$
avoid $y,w$, and if $x,w$ avoid $y,z$.

\begin{lemma}\label{lem:bqs}
If $G$ is a circular-arc graph, then $G$ has no blocking quadruple.
\end{lemma}

To see this, observe that a BQ is always an independent set of size four.
Now, suppose that $G$ has a circular-arc representation and the arcs
representing vertices $x,y,z,w$ appear in this circular order. Then no path
between $x$ and $z$ can be missed by both $y$ and $w$, and no path between $y$
and $w$ can be missed~by~$x$~and~$z$. In other words, the vertices $x, z$ do not
avoid $y, w$.

\begin{figure}[!b]
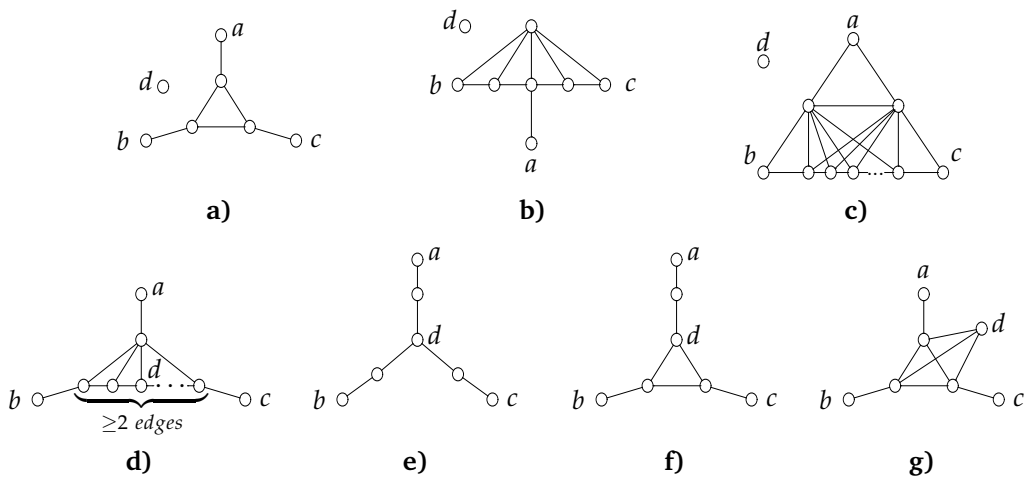
\centering
\qquad\begin{tabular}{@{}c@{\qquad\qquad}c@{\qquad\qquad}c@{}}
$\xy/r1.8pc/: (0,0.8)*[o][F]{\phantom{s}}="x1"; (1,0.8)*[o][F]{\phantom{s}}="x2";
(0.5,1.6)*[o][F]{\phantom{s}}="x3"; (0.5,2.4)*[o][F]{\phantom{s}}="x4";
(-0.8,0.56)*[o][F]{\phantom{s}}="x5"; (1.8,0.56)*[o][F]{\phantom{s}}="x6";
(-0.5,1.5)*[o][F]{\phantom{s}}="x7"; (0,0)*{}; {\ar@{-} "x1";"x2"};
{\ar@{-} "x1";"x3"};{\ar@{-} "x2";"x3"};{\ar@{-} "x3";"x4"};{\ar@{-} "x1";"x5"};
{\ar@{-} "x2";"x6"}; "x4"+(0.3,0.1)*{a}; "x5"+(-0.4,0)*{b}; "x6"+(0.35,0)*{c};
"x7"+(-0.3,0.1)*{d}; \endxy$
& 
$\xy/r2.3pc/: (-0.5,1.2)*[o][F]{\phantom{s}}="x1";
(1.5,1.2)*[o][F]{\phantom{s}}="x2"; (0.5,2.0)*[o][F]{\phantom{s}}="x3";
(0.5,0.4)*[o][F]{\phantom{s}}="x4"; (0,1.2)*[o][F]{\phantom{s}}="x71";
(0.5,1.2)*[o][F]{\phantom{s}}="x72"; (1,1.2)*[o][F]{\phantom{s}}="x73";
(-0.4,2.0)*[o][F]{\phantom{s}}="x5"; (0,0)*{}; {\ar@{-} "x1";"x71"};
{\ar@{-} "x71";"x72"}; {\ar@{-} "x72";"x73"}; {\ar@{-} "x73";"x2"};
{\ar@{-} "x1";"x3"}; {\ar@{-} "x2";"x3"}; {\ar@{-} "x72";"x4"};
{\ar@{-} "x3";"x71"}; {\ar@{-} "x3";"x72"}; {\ar@{-} "x3";"x73"};
"x4"+(0.0,-0.3)*{a}; "x1"+(-0.3,0)*{b}; "x2"+(0.35,0)*{c};
"x5"+(-0.2,0.1)*{d}; \endxy$
& 
$\xy/r1.4pc/: (0,3)*[o][F]{\phantom{s}}="1"; (-2,0)*[o][F]{\phantom{s}}="2";
(2,0)*[o][F]{\phantom{s}}="3"; (-1,1.5)*[o][F]{\phantom{s}}="4";
(-1,0)*[o][F]{\phantom{s}}="51"; (-0.5,0)*[o][F]{\phantom{s}}="52";
(0,0)*[o][F]{\phantom{s}}="53";(0.5,0)*{_{\ldots}};(1,0)*[o][F]{\phantom{s}}="54";
(1,1.5)*[o][F]{\phantom{s}}="6";(-2,2.5)*[o][F]{\phantom{s}}="7";
{\ar@{-} "1";"4"}; {\ar@{-} "1";"6"}; {\ar@{-} "2";"4"}; {\ar@{-} "2";"51"};
{\ar@{-} "3";"54"}; {\ar@{-} "3";"6"}; {\ar@{-} "4";"51"}; {\ar@{-} "4";"52"};
{\ar@{-} "4";"53"}; {\ar@{-} "4";"54"}; {\ar@{-} "4";"6"}; {\ar@{-} "51";"6"};
{\ar@{-} "52";"6"}; {\ar@{-} "53";"6"}; {\ar@{-} "54";"6"}; {\ar@{-} "51";"52"};
{\ar@{-} "52";"53"}; {\ar@{-} "53";"53"+(0.3,0)}; {\ar@{-} "54";"54"-(0.3,0)};
"1"+(0,0.4)*{a}; "2"+(-0.3,0.4)*{b}; "3"+(0.3,0.4)*{c}; "7"+(0,0.45)*{d};\endxy$
\smallskip\\
\bf a) & \bf b) & \bf c)
\end{tabular}\medskip

\begin{tabular}{@{}c@{\quad~~}c@{\qquad}c@{\quad~~}c@{}}
$\xy/r1.8pc/: (-0.5,0.8)*[o][F]{\phantom{s}}="x1";
(1.5,0.8)*[o][F]{\phantom{s}}="x2"; (0.5,1.6)*[o][F]{\phantom{s}}="x3";
(0.5,2.4)*[o][F]{\phantom{s}}="x4"; (-1.3,0.56)*[o][F]{\phantom{s}}="x5";
(2.3,0.56)*[o][F]{\phantom{s}}="x6"; (0,0.8)*[o][F]{\phantom{s}}="x71";
(0.5,0.8)*[o][F]{\phantom{s}}="x72"; (1,0.8)*{\ldots}; (0,0)*{};
{\ar@{-} "x1";"x71"}; {\ar@{-} "x71";"x72"}; {\ar@{-} "x72";"x72"+(0.2,0)};
{\ar@{-} "x2";"x2"-(0.2,0)}; {\ar@{-} "x1";"x3"}; {\ar@{-} "x2";"x3"};
{\ar@{-} "x3";"x4"}; {\ar@{-} "x1";"x5"}; {\ar@{-} "x2";"x6"};
{\ar@{-} "x3";"x71"}; {\ar@{-} "x3";"x72"}; "x4"+(0.3,0.1)*{a};
"x5"+(-0.4,0)*{b}; "x6"+(0.35,0)*{c}; "x72"+(0.23,0.3)*{d};
"x72"+(0,-0.5)*{\underbrace{\hspace{5em}}_{\geq 2~edges}};\endxy$
& 
$\xy/r1.8pc/: (-0.2,1.0)*[o][F]{\phantom{s}}="x1";
(1.2,1.0)*[o][F]{\phantom{s}}="x2"; (0.5,1.6)*[o][F]{\phantom{s}}="x3";
(0.5,2.4)*[o][F]{\phantom{s}}="x4"; (-0.8,0.56)*[o][F]{\phantom{s}}="x5";
(1.8,0.56)*[o][F]{\phantom{s}}="x6"; (0.5,3.0)*[o][F]{\phantom{s}}="x7";
(0,0)*{}; {\ar@{-} "x1";"x3"}; {\ar@{-} "x2";"x3"}; {\ar@{-} "x3";"x4"};
{\ar@{-} "x1";"x5"}; {\ar@{-} "x2";"x6"};{\ar@{-} "x4";"x7"};"x3"+(0.3,0.1)*{d};
"x5"+(-0.4,0)*{b}; "x6"+(0.35,0)*{c}; "x7"+(0.3,0.1)*{a}; \endxy$
& 
$\xy/r1.8pc/: (0,0.8)*[o][F]{\phantom{s}}="x1";(1,0.8)*[o][F]{\phantom{s}}="x2";
(0.5,1.6)*[o][F]{\phantom{s}}="x3"; (0.5,2.4)*[o][F]{\phantom{s}}="x4";
(-0.8,0.56)*[o][F]{\phantom{s}}="x5"; (1.8,0.56)*[o][F]{\phantom{s}}="x6";
(0.5,3.0)*[o][F]{\phantom{s}}="x7"; (0,0)*{}; {\ar@{-} "x1";"x2"};
{\ar@{-} "x1";"x3"};{\ar@{-} "x2";"x3"};{\ar@{-} "x3";"x4"};{\ar@{-} "x1";"x5"};
{\ar@{-} "x2";"x6"}; {\ar@{-} "x4";"x7"}; "x3"+(0.3,0.1)*{d}; "x5"+(-0.4,0)*{b};
"x6"+(0.35,0)*{c}; "x7"+(0.3,0.1)*{a}; \endxy$
& 
$\xy/r1.8pc/: (0,0.8)*[o][F]{\phantom{s}}="x1";(1,0.8)*[o][F]{\phantom{s}}="x2";
(0.5,1.6)*[o][F]{\phantom{s}}="x3"; (0.5,2.4)*[o][F]{\phantom{s}}="x4";
(-0.8,0.56)*[o][F]{\phantom{s}}="x5"; (1.8,0.56)*[o][F]{\phantom{s}}="x6";
(1.5,1.8)*[o][F]{\phantom{s}}="x7"; (0,0)*{}; {\ar@{-} "x1";"x2"};
{\ar@{-} "x1";"x3"};{\ar@{-} "x2";"x3"};{\ar@{-} "x3";"x4"};{\ar@{-} "x1";"x5"};
{\ar@{-} "x2";"x6"};{\ar@{-} "x1";"x7"};{\ar@{-} "x2";"x7"};{\ar@{-} "x3";"x7"};
"x4"+(0,0.4)*{a}; "x5"+(-0.4,0)*{b}; "x6"+(0.35,0)*{c};
"x7"+(0.3,0.1)*{d};\endxy$
\\
\bf d) & \bf e) & \bf f) & \bf g)
\end{tabular}
\caption{Forbidden induced subgraph characterization of chordal graphs with no BQs.\label{fig:forb}}
\end{figure}

Let us now discuss various forms of blocking quadruples that one may encounter
in graphs.  One class of such examples arises from asteroidal triples:  if
$a,b,c$ form an asteroidal triple of $G$ and $d$ is a vertex of degree zero
in $G$, then $a,b,c,d$ is a blocking quadruple. This can be seen in the first
three graphs in Figure~\ref{fig:forb}. Other ways of extending an asteroidal
triple to a BQ are also illustrated in Figure~\ref{fig:forb}.  The vertices
$a,b,c$ in each of the graphs in the second row form an asteroidal triple while
the vertices $a,b,c,d$ form a blocking quadruple.  For chordal graphs, these are
all possible forms of BQs (see Theorem~\ref{thm:bq-chordal}).

\begin{figure}[h!]
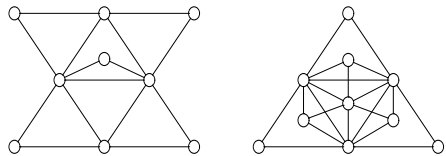
\centering
$\xy/r1.4pc/: (0,3)*[o][F]{\phantom{s}}="1"; (-2,0)*[o][F]{\phantom{s}}="2";
(-2,3)*[o][F]{\phantom{s}}="2'"; (2,0)*[o][F]{\phantom{s}}="3";
(2,3)*[o][F]{\phantom{s}}="3'"; (-1,1.5)*[o][F]{\phantom{s}}="4";
(0,0)*[o][F]{\phantom{s}}="5"; (1,1.5)*[o][F]{\phantom{s}}="6";
(0,1.975)*[o][F]{\phantom{s}}="7"; {\ar@{-} "1";"4"}; {\ar@{-} "1";"6"};
{\ar@{-} "2";"4"}; {\ar@{-} "2'";"4"}; {\ar@{-} "2'";"1"}; {\ar@{-} "3'";"6"};
{\ar@{-} "3'";"1"}; {\ar@{-} "2";"5"}; {\ar@{-} "3";"5"}; {\ar@{-} "3";"6"};
{\ar@{-} "4";"5"}; {\ar@{-} "4";"6"}; {\ar@{-} "5";"6"}; {\ar@{-} "7";"4"};
{\ar@{-} "7";"6"};\endxy$\qquad
$\xy/r1.4pc/: (0,3)*[o][F]{\phantom{s}}="1"; (-2,0)*[o][F]{\phantom{s}}="2";
(2,0)*[o][F]{\phantom{s}}="3"; (-1,1.5)*[o][F]{\phantom{s}}="4";
(0,0)*[o][F]{\phantom{s}}="5"; (1,1.5)*[o][F]{\phantom{s}}="6";
(0,0.975)*[o][F]{\phantom{s}}="7"; (0,1.95)*[o][F]{\phantom{s}}="z";
(-1,0.6)*[o][F]{\phantom{s}}="x"; (1,0.6)*[o][F]{\phantom{s}}="y";
{\ar@{-} "1";"4"}; {\ar@{-} "1";"6"}; {\ar@{-} "2";"4"}; {\ar@{-} "2";"5"};
{\ar@{-} "3";"5"}; {\ar@{-} "3";"6"}; {\ar@{-} "4";"5"}; {\ar@{-} "4";"6"};
{\ar@{-} "5";"6"}; {\ar@{-} "7";"4"}; {\ar@{-} "7";"5"}; {\ar@{-} "7";"6"};
{\ar@{-} "7";"x"}; {\ar@{-} "7";"y"}; {\ar@{-} "7";"z"}; {\ar@{-} "z";"4"};
{\ar@{-} "z";"6"}; {\ar@{-} "x";"4"}; {\ar@{-} "x";"5"}; {\ar@{-} "y";"5"};
{\ar@{-} "y";"6"};\endxy$
\caption{Some minimal chordal non-circular-arc graphs with no BQs.\label{fig:no-bqs}}
\end{figure}

Unlike these examples, the two chordal graphs in Figure~\ref{fig:no-bqs} do not
contain blocking quadruples, and yet they are not circular-arc graphs.  Thus the
absence of blocking quadruples is not sufficient to guarantee that a (chordal)
graph is a circular-arc graph.  However, in some cases, it may be sufficient.

For instance, a result of \cite{bon09} (Corollary 15) can be restated as
follows.

\begin{lemma}
A claw-free chordal graph is a circular-arc graph iff it has no BQ.
\end{lemma}

We prove a similar statement for chordal graphs of independence number at most
four (see Theorem~\ref{thm:main}).  The absence of BQs therefore gives us a
simple and uniform forbidden structure characterization of these classes, as
opposed to more common forbidden induced subgraph characterizations
\cite{bang94,bon09,lin06,tuck74}.

\section{Main results}\label{sec:main}

\noindent In this section, we summarize the main theorems of this paper.

In the first theorem, we describe all minimal forbidden induced subgraphs characterizing
chordal graphs with no BQs.  These are the graphs depicted in Figure
\ref{fig:forb}.

\begin{theorem}\label{thm:bq-chordal}
If $G$ is chordal, then the following are equivalent.
\begin{compactenum}[(i)]
\item $G$ contains a blocking quadruple.
\item $G$ contains an induced subgraph isomorphic to a graph in Figure~\ref{fig:forb}.
\end{compactenum}
\end{theorem}
\noindent In fact, the theorem holds for the more general class of {\em nearly
chordal} graphs (a graph class defined in \cite{brandstadt} generalizing both chordal
and circular-arc graphs).

In the second theorem, we show that the absence of BQs is necessary and sufficient for a chordal
graph of independence number $\alpha(G)\leq 4$ to be a circular-arc graph.

\begin{theorem}\label{thm:main}
If $G$ is chordal and $\alpha(G)\leq 4$, the following are equivalent.
\begin{compactenum}[(i)]
\item $G$ is a circular-arc graph.
\item $G$ contains no blocking quadruple.
\end{compactenum}
\end{theorem}

\noindent The theorem fails for chordal graphs $G$ with $\alpha(G)\geq 5$ as
Figure~\ref{fig:no-bqs} shows.

Notice that this theorem implies that every chordal graph $G$ with
$\alpha(G)\leq 3$ is a circular-arc graph (as any blocking quadruple is
necessarily an independent set of size four).  In contrast, it is known that
every chordal graph $G$ with $\alpha(G)\leq 2$ is an interval graph (an
asteroidal triple is also an independent set), and in fact, a proper interval
graph (claw contains an independent set of size three).

\section{Proof of Theorem~\ref{thm:bq-chordal}}

\noindent We prove a slightly more general statement (see Theorem~\ref{thm:bq-nearly} below).

A graph $G$ is a {\em nearly chordal} graph \cite{brandstadt} if for each $v\in
V(G)$, the graph $G-N[v]$ is a chordal graph. A~graph $G$ is a {\em nearly
interval} graph if for each $v\in V(G)$, the graph $G-N[v]$ is an interval
graph.  Clearly, every chordal graph is nearly chordal.

\begin{theorem}\label{thm:bq-nearly}
If $G$ is nearly chordal, then the following are equivalent.
\begin{compactenum}[(i)]
\item $G$ contains a blocking quadruple.
\item $G$ contains an induced subgraph isomorphic to a graph in Figure~\ref{fig:forb}.
\end{compactenum}
\end{theorem}

The proof of this theorem is split into the following two claims.

\begin{lemma}\label{lem:forb1}
Let $G$ be a nearly chordal graph. If $G$ is not a nearly interval graph, then
$G$ contains one of the graphs in Figure~\ref{fig:forb} as an induced subgraph.
\end{lemma}

\begin{proof}
Suppose that $G$ is not a nearly interval graph. Then there exists a vertex
$v\in V(G)$ such that $G-N[v]$ is not an interval graph. Since $G$ is a nearly
chordal graph, we have that $G-N[v]$ is a chordal graph. Thus $G-N[v]$ contains
as an induced subgraph one of the graphs a)-e) shown in Figure
\ref{fig:interval} \cite{lekker62}.

If $G-N[v]$ contains the graphs d) or e) from Figure~\ref{fig:interval}, then so
does $G$ and the two graphs are also in Figure~\ref{fig:forb}, namely d) and e),
respectively.  So we may assume that $G-N[v]$ contains one of the graphs a)-c)
shown in Figure~\ref{fig:interval}, which together with $v$ yields an induced
subgraph of $G$ that is one of  the graphs a)-c) shown in Figure~\ref{fig:forb}.
This concludes the proof.
\end{proof}

\begin{figure}[!ht]
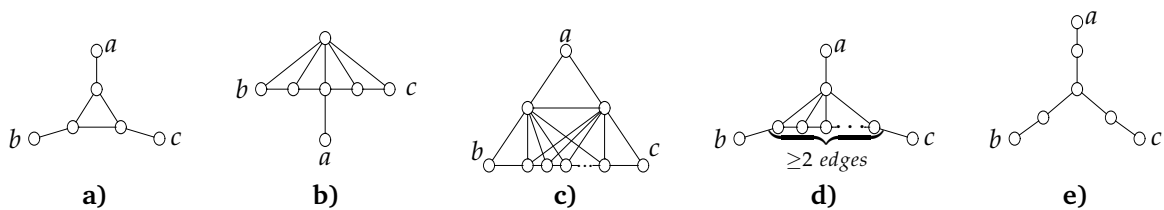
\centering
\begin{tabular}{@{}c@{\qquad}c@{\qquad}c@{\qquad}c@{\qquad}c@{}}
$\xy/r1.5pc/: (0,0.8)*[o][F]{\phantom{s}}="x1"; (1,0.8)*[o][F]{\phantom{s}}="x2";
(0.5,1.6)*[o][F]{\phantom{s}}="x3"; (0.5,2.4)*[o][F]{\phantom{s}}="x4";
(-0.8,0.56)*[o][F]{\phantom{s}}="x5"; (1.8,0.56)*[o][F]{\phantom{s}}="x6";
(0,0)*{}; {\ar@{-} "x1";"x2"}; {\ar@{-} "x1";"x3"}; {\ar@{-} "x2";"x3"};
{\ar@{-} "x3";"x4"}; {\ar@{-} "x1";"x5"}; {\ar@{-} "x2";"x6"};
"x4"+(0.3,0.1)*{a}; "x5"+(-0.4,0)*{b}; "x6"+(0.35,0)*{c}; \endxy$
& 
$\xy/r2.0pc/: (-0.5,1.2)*[o][F]{\phantom{s}}="x1";
(1.5,1.2)*[o][F]{\phantom{s}}="x2"; (0.5,2.0)*[o][F]{\phantom{s}}="x3";
(0.5,0.4)*[o][F]{\phantom{s}}="x4"; (0,1.2)*[o][F]{\phantom{s}}="x71";
(0.5,1.2)*[o][F]{\phantom{s}}="x72"; (1,1.2)*[o][F]{\phantom{s}}="x73";
(0,0)*{}; {\ar@{-} "x1";"x71"}; {\ar@{-} "x71";"x72"}; {\ar@{-} "x72";"x73"};
{\ar@{-} "x73";"x2"}; {\ar@{-} "x1";"x3"}; {\ar@{-} "x2";"x3"};
{\ar@{-} "x72";"x4"}; {\ar@{-} "x3";"x71"}; {\ar@{-} "x3";"x72"};
{\ar@{-} "x3";"x73"}; "x4"+(0.0,-0.3)*{a}; "x1"+(-0.3,0)*{b};
"x2"+(0.35,0)*{c};\endxy$
& 
$\xy/r1.2pc/: (0,3)*[o][F]{\phantom{s}}="1"; (-2,0)*[o][F]{\phantom{s}}="2";
(2,0)*[o][F]{\phantom{s}}="3"; (-1,1.5)*[o][F]{\phantom{s}}="4";
(-1,0)*[o][F]{\phantom{s}}="51"; (-0.5,0)*[o][F]{\phantom{s}}="52";
(0,0)*[o][F]{\phantom{s}}="53"; (0.5,0)*{_{\ldots}};
(1,0)*[o][F]{\phantom{s}}="54"; (1,1.5)*[o][F]{\phantom{s}}="6";
{\ar@{-} "1";"4"}; {\ar@{-} "1";"6"}; {\ar@{-} "2";"4"}; {\ar@{-} "2";"51"};
{\ar@{-} "3";"54"}; {\ar@{-} "3";"6"}; {\ar@{-} "4";"51"}; {\ar@{-} "4";"52"};
{\ar@{-} "4";"53"}; {\ar@{-} "4";"54"}; {\ar@{-} "4";"6"}; {\ar@{-} "51";"6"};
{\ar@{-} "52";"6"}; {\ar@{-} "53";"6"}; {\ar@{-} "54";"6"}; {\ar@{-} "51";"52"};
{\ar@{-} "52";"53"}; {\ar@{-} "53";"53"+(0.3,0)}; {\ar@{-} "54";"54"-(0.3,0)};
"1"+(0,0.4)*{a}; "2"+(-0.3,0.4)*{b}; "3"+(0.3,0.4)*{c}; \endxy$
& 
$\xy/r1.5pc/: (-0.5,0.8)*[o][F]{\phantom{s}}="x1";
(1.5,0.8)*[o][F]{\phantom{s}}="x2"; (0.5,1.6)*[o][F]{\phantom{s}}="x3";
(0.5,2.4)*[o][F]{\phantom{s}}="x4"; (-1.3,0.56)*[o][F]{\phantom{s}}="x5";
(2.3,0.56)*[o][F]{\phantom{s}}="x6"; (0,0.8)*[o][F]{\phantom{s}}="x71";
(0.5,0.8)*[o][F]{\phantom{s}}="x72"; (1,0.8)*{\ldots}; (0,0)*{};
{\ar@{-} "x1";"x71"}; {\ar@{-} "x71";"x72"}; {\ar@{-} "x72";"x72"+(0.2,0)};
{\ar@{-} "x2";"x2"-(0.2,0)}; {\ar@{-} "x1";"x3"}; {\ar@{-} "x2";"x3"};
{\ar@{-} "x3";"x4"}; {\ar@{-} "x1";"x5"}; {\ar@{-} "x2";"x6"};
{\ar@{-} "x3";"x71"}; {\ar@{-} "x3";"x72"};
"x72"+(0,-0.5)*{\underbrace{\hspace{4.3em}}_{\geq 2~edges}}; "x4"+(0.3,0.1)*{a};
"x5"+(-0.4,0)*{b}; "x6"+(0.35,0)*{c}; \endxy$
& 
$\xy/r1.5pc/: (-0.2,1.0)*[o][F]{\phantom{s}}="x1";
(1.2,1.0)*[o][F]{\phantom{s}}="x2"; (0.5,1.6)*[o][F]{\phantom{s}}="x3";
(0.5,2.4)*[o][F]{\phantom{s}}="x4"; (-0.8,0.56)*[o][F]{\phantom{s}}="x5";
(1.8,0.56)*[o][F]{\phantom{s}}="x6"; (0.5,3.0)*[o][F]{\phantom{s}}="x7";
(0,0)*{}; {\ar@{-} "x1";"x3"}; {\ar@{-} "x2";"x3"}; {\ar@{-} "x3";"x4"};
{\ar@{-} "x1";"x5"}; {\ar@{-} "x2";"x6"}; {\ar@{-} "x4";"x7"};
"x5"+(-0.4,0)*{b}; "x6"+(0.35,0)*{c}; "x7"+(0.3,0.1)*{a};\endxy$
\\ \bf a) & \bf b) & \bf c) & \bf d) & \bf e)
\end{tabular}
\caption{All chordal minimal forbidden induced subgraphs of interval
graphs.\label{fig:interval}}
\end{figure} 

\begin{lemma}\label{lem:forb2}
Let $G$ be a nearly interval graph. If $G$ has a blocking quadruple, then
$G$ contains one of the graphs in Figure~\ref{fig:forb} as an induced subgraph.
\end{lemma}

\begin{proof}
Suppose that $G$ contains a blocking quadruple $a,b,c,d$.  In particular, this
implies that $G$ contains an $ab$-path missed by both $c$ and $d$, or a
$cd$-path missed by both $a$ and $b$. Similarly, we have that $G$ contains an
$ac$-path missed by $b,d$, or a $bd$-path missed by $a,c$.

By symmetry between $a,b,c,d$, we may assume without loss of generality that $G$
contains an $ab$-path $P_b$ missed by $c,d$ and an $ac$-path $P_c$ missed by
$b,d$. (If not we replace $a$ by one of $b,c,d$ to satisfy this.) 
If $G$ also contains a $bc$-path missed by $a$ and $d$, then the vertices
$a,b,c$ form an asteroidal triple in $G-N[d]$. Thus $G-N[d]$ is not an interval
graph, contradicting our assumption that $G$ is a nearly interval graph.

We therefore conclude that every $bc$-path of $G$ contains a neighbour of $a$ or
$d$. Since $a,b,c,d$ is a blocking quadruple, this implies that $G$ contains an
$ad$-path $P_d$ missed by both $b$ and $c$.  Moreover, since $G$ is a nearly
interval graph, we have that every $bd$-path contains a neighbour of $a$ or
$c$, and every $cd$-path contains a neighbour of $a$ or $b$. (Otherwise
$G-N[c]$ or $G-N[b]$ is not an interval graph, providing a similar contradiction
as above.) For later use, we summarize these observations as follows.

\begin{claim}\label{clm:1}
Every $bc$-path of $G$ contains a neighbour of $a$ or $d$, every $bd$-path
contains a neighbour of $a$ or $c$, and every $cd$-path contains a neighbour of
$a$ or $b$.
\end{claim}

Now, without loss of generality, we shall assume that the quadruple $a,b,c,d$
and the paths $P_b$, $P_c$, $P_d$ were chosen so that $|P_b|+|P_c|+|P_d|$ is
smallest possible. (If not, we simply replace $a,b,c,d$ or paths $P_b,P_c,P_d$
by a different quadruple or paths where the value is smaller and repeat
the above.)  The minimality of this choice implies that $P_b$ is a shortest
$ab$-path in $G-N[c]-N[d]$. Similarly, $P_c$ is a shortest $ac$-path in
$G-N[b]-N[d]$, and $P_d$ is a shortest $ad$-path in $G-N[b]-N[c]$.  In
particular, the three paths are induced paths of $G$.

Observe that the paths $P_b$, $P_c$, $P_d$ share at least one vertex, namely
$a$. We prove that the minimality of our choice guarantees that $a$ is, in fact,
the only vertex that these paths share. This is proved as follows.

\begin{claim}\label{clm:2}
$P_b\cap P_c\cap P_d=\{a\}$. 
\end{claim}

\begin{proofclaim}
Suppose that $P_b\cap P_c\cap P_d\neq\{a\}$.  Let $x$ be the vertex on $P_b$
such that $x\in P_c\cap P_d$ and no internal vertex of the subpath of $P_b$
between $x$ and $b$ is in $P_c\cap P_d$. Clearly, $x$ is well-defined, since
$a\in P_b\cap P_c\cap P_d$.  Moreover, we have $x\neq a$ by the choice of
$x$ and the fact that $P_b\cap P_c\cap P_d\neq\{a\}$.  Denote by $P_b'$ the
subpath of $P_b$ between $x$ and $b$. Similarly, since $x\in P_c\cap P_d$,
denote by $P'_c$ the subpath of $P_c$ between $x$ and $c$, and by $P'_d$ the
subpath of $P_d$ between $x$ and $d$.  Observe that $P'_b$ is missed by both $c$
and $d$, since $P_b$ is.  Similary, $P'_c$ is missed by $b,d$ and $P'_d$ is
missed by $b,c$, since $P_c$, respectively, $P_d$ is. We deduce that $x,b,c,d$
is a blocking quadruple of $G$, and $|P'_b|+|P'_c|+|P'_d| < |P_b|+|P_c|+|P_d|$,
since $x\neq a$.  But that contradicts the choice of $a,b,c,d$.

\end{proofclaim}

\begin{claim}\label{clm:3}
$P_b\cap P_c=P_b\cap P_d=P_c\cap P_d=\{a\}$.
\end{claim}

\begin{proofclaim}
By symmetry, suppose without loss of generality that $P_b\cap (P_c\cup
P_d)\neq\{a\}$. Recall that $a\in P_b\cap P_c\cap P_d$.  Let $x$ be the vertex
of $P_b$ such that $x\in P_c$ and no internal vertex of the subpath of $P_b$
between $x$ and $b$ is in $P_c$.  Let $y$ be the vertex of $P_b$ such
that $y\in P_d$ and no internal vertex of the subpath of $P_b$ between $y$
and $b$~is~in~$P_d$. 

Suppose that $x=y$.  If $x=y=a$, then we conclude $P_b\cap P_c=\{a\}$  and
$P_b\cap P_d=\{a\}$ by the choice of $x$ and $y$. However, we assume that
$P_b\cap (P_c\cup P_d)\neq\{a\}$.  Thus $x=y\neq a$ in which case we contradict
\ref{clm:2}, since $x\in P_b\cap P_c$ and $y\in P_b\cap P_d$.  Therefore, we
must conclude that $x\neq y$.

Now, by symmetry, we shall assume without loss of generality that $x$ belongs to
the subpath of $P_b$ between $y$ and $b$.  Denote by $P'_b$ the subpath of $P_b$
between $x$ and $b$, and by $P^{yx}_b$ the subpath of $P_b$ between $y$ and $x$.
Since $x\in P_c$, denote by $P'_c$ the subpath of $P_c$ between $x$ and $c$.
Also, since $y\in P_d$, denote by $P^y_d$ the subpath of $P_d$ between $d$ and
$y$, and define $P'_d=P^y_d\cup P^{yx}_b$.  Note that $P'_d$ is a path between
$d$ and $x$, since $P^{yx}_b\cap P_d=\{y\}$ by the choice of $y$.
Further, observe that both $P'_b$ and $P^{yx}_b$ are missed by $c$ and $d$, since
$P_b$ is. Similarly, $P'_c$ is missed by both $b$ and $d$, since $P_c$ is.
Further, note that $xb\not\in E(G)$, since $x\in P_c$ and $P_c$ is missed by
$b$.  This implies that $P^{yx}_b$ is missed by $b$, since $P_b$ is an induced
path.  Therefore, $P'_d$ is missed by both $b$ and $c$, since both $P_d$ and
$P^{yx}_b$ are.
This shows that $x,b,c,d$ is a blocking quadruple. However, we observe that
$a\neq x$, since $x\neq y$ and $x$ belongs to the subpath of $P_b$ between $y$
and $b$. 

Thus $|P_c'|<|P_c|$ which implies $|P'_b|+|P'_c|+|P'_d| =
|P'_b|+|P'_c|+|P^{yx}_b|+|P^y_d|-1 \leq |P_b|+|P'_c|+|P_d| < |P_b|+|P_c|+|P_d|$
which contradicts our choice of the quadruple $a,b,c,d$.
\end{proofclaim}

In the next claims, we analyze the edges between the paths $P_b$, $P_c$, $P_d$.

\begin{claim}\label{clm:4}
For every edge $uv$ of $G$ such that $u,v\in P_b\cup P_c\cup P_d$, we have that
$u\in N(a)$ or $v\in N(a)$ or $uv$ is an edge of one of the paths $P_b,P_c,P_d$.
\end{claim}

\begin{proofclaim}
Consider an edge $uv\in E(G)$ where $u,v\in P_b\cup P_c\cup P_d$.  For
contradiction, suppose that $u,v\not\in N(a)$ and that $uv$ is not an edge of
one of the paths $P_b,P_c,P_d$. By symmetry, we may assume without loss of
generality that $u\in P_b$ and $v\in P_c$.  Let $P^u_b$ denote the subpath of
$P_b$ between $b$ and $u$, and let $P^v_c$ denote the subpath of $P_c$ between
$v$ and $c$.  Note that both $P^u_b$ and $P^v_c$ are missed by $a$, since
$u,v\not\in N(a)$ and both $P_b$ and $P_c$ are induced paths.
Define $P'=P^u_b\cup P^v_c$, and observe that $P'$ is a path by \ref{clm:3}.
Moreover, $P'$ is missed by $d$, since both $P_b$ and $P_c$ are. Also, $P'$ is
missed by $a$, since both $P^u_b$ and $P^v_c$ are.  But then $P'$ is a $bc$-path
missed by both $a$ and~$d$, contradicting \ref{clm:1}.
\end{proofclaim}

\begin{claim}\label{clm:7b}
Suppose that $u\in P_c\cup P_d$ has a neighbour $v$ on $P_b$. Then $u$ is
adjacent to every internal vertex of the subpath of $P_b$ between $a$ and $v$.
\end{claim}

\begin{proofclaim}
By symmetry (replacing $c$ by $d$ and vice-versa), we may assume without loss of
generality that $u\in P_c$.

For contradiction, suppose that $u$ is not adjacent to some internal vertex of
the subpath of $P_b$ between $a$ and $v$. This implies that $v\not\in N(a)$, and
hence, $u\in N(a)$ by \ref{clm:4}. From this we conclude that $u$ has
distinct neighbours $x,y$ on $P_b$ (possibly $\{x,y\}\cap \{a,v\}\neq\emptyset$) such
that $x,y$ are not consecutive on $P_b$ and no internal vertex of the subpath of
$P_b$ between $x$ and $y$ is adjacent to $u$. Let $P^{xy}_b$ denote this
subpath.  Note that $P^{xy}_b$ is an induced path, since $P_b$ is. So, since
$x,y$ are not consecutive on $P_b$, we have $|P^{xy}_b|\geq 3$.  Further, recall
that the only neighbours of $u$ on $P^{xy}_b$ are $x$ and $y$. From this we
conclude that $C=P^{xy}_b\cup\{u\}$ induces in $G$ a cycle of length four or
more. In particular, no vertex of $C$ is adjacent to $d$, since $C\subseteq
P_b\cup P_c$ and both $P_b$ and $P_c$ are missed by $d$.  Thus $C$ is an induced
cycle of length four or more in $G-N[d]$, contradicting our assumption that $G$
is a nearly interval graph.
\end{proofclaim}

Using the symmetry between the paths $P_b$, $P_c$, $P_d$, we may generalize this
to the following statement.

\begin{claimplus}\label{clm:7}
Suppose that $u\in P_c\cup P_d$ has a neighbour $v$ on $P_b$. Then $u$ is
adjacent to every internal vertex of the subpath of $P_b$ between $a$ and $v$.

Suppose that $u\in P_b\cup P_d$ has a neighbour $v$ on $P_c$. Then $u$ is
adjacent to every internal vertex of the subpath of $P_c$ between $a$ and $v$.

Suppose that $u\in P_b\cup P_c$ has a neighbour $v$ on $P_d$. Then $u$ is
adjacent to every internal vertex of the subpath of $P_d$ between $a$ and $v$.
\end{claimplus}

\begin{claim}\label{clm:5b}
Suppose that $|P_b|\geq 4$. If there are edges between the internal vertices of
$P_b$ and $P_c$, then there are no edges between the internal
vertices of $P_b$ and $P_d$.
\end{claim}

\begin{proofclaim}
For contradiction, suppose that some internal vertex $v$ of $P_b$ has a
neighbour $x\not\in\{a,c\}$ on $P_c$, and some internal vertex $w$ of $P_b$
(possibly $v=w$) has a neighbour $y\not\in\{a,d\}$ on $P_d$.  Let $u$ denote the
neighbour of $a$ on $P_b$. From \ref{clm:7} applied to $x$ and $v$, we conclude
that $ux\in E(G)$.  Similarly, using \ref{clm:7} applied to $y$~and~$w$, we
obtain $uy\in E(G)$.

Now, recall that $|P_b|\geq 4$.  This implies that $ub\not\in E(G)$, since $P_b$
is an induced path.  Let $P^u_b$ be the subpath of $P_b$ between $u$ and $b$,
let $P^x_c$ be the subpath of $P_c$ between $x$ and $c$, and let $P^y_d$ be the
subpath of $P_d$ between $y$ and $d$. 
Define $P'_b=P^u_b$, $P'_c=\{u\}\cup P^x_c$, and $P'_d=\{u\}\cup P^y_d$.
Clearly, both $P'_c$ and $P'_d$ are paths of $G$, since $u\not\in P_c\cup P_d$
by \ref{clm:3}.  Moreover, $uc,ud\not\in E(G)$ because $u\in P_b$ and $P_b$ is
missed by both $c$ and $d$.

Observe now that $P'_b$ is missed by $c$ and $d$, since $P_b$ is.  Moreover,
$P'_c$ is missed by both $b$ and $d$, since $ub,ud\not\in E(G)$ and since $P_c$
is missed by both $b$ and $d$.  Similarly, $P'_d$ is missed by both $b$ and $c$,
since $ub,uc\not\in E(G)$ and since $P_d$ is missed by both $b$ and $c$.  Thus
we conclude that $u,b,c,d$ is a blocking quadruple of $G$. In particular, we have
$|P'_c|\leq |P_c|$ and $|P'_d|\leq |P_d|$, since $a\not\in\{x,y\}$. Also,
$|P'_b|<|P_b|$ by the definition of $u$. So we conclude that
$|P'_b|+|P'_c|+|P'_d| < |P_b|+ |P_c|+|P_d|$, which contradicts the minimality of
our choice of $a,b,c,d$ and paths $P_b$, $P_c$, $P_d$.
\end{proofclaim}

A symmetric argument (using $P_c$ or $P_d$ in place of $P_b$) yields the following.

\begin{claimplus}\label{clm:5}
Suppose that $|P_b|\geq 4$. If there are edges between the internal vertices of
$P_b$ and $P_c$, then there are no edges between the internal
vertices of $P_b$ and $P_d$.

Suppose that $|P_c|\geq 4$. If there are edges between the internal vertices of
$P_c$ and $P_b$, then there are no edges between the internal
vertices of $P_c$ and $P_d$.

Suppose that $|P_d|\geq 4$. If there are edges between the internal vertices of
$P_d$ and $P_b$, then there are no edges between the internal
vertices of $P_d$ and $P_c$.
\end{claimplus}

\begin{claim}\label{clm:6b}
Suppose that $|P_b|\geq 4$ and $|P_c|\geq 4$. Then there are no edges between
the internal vertices of $P_b$ and $P_c$.
\end{claim}

\begin{proofclaim}
For contradiction, suppose that an internal vertex $u$ of $P_b$ is adjacent to
an internal vertex $v$ of $P_c$.  By \ref{clm:4}, we conclude $u\in N(a)$ or
$v\in N(a)$.  By symmetry, we may assume without loss of generality that $u\in
N(a)$. This implies $ub\not\in E(G)$, since $P_b$ is an induced path and
$|P_b|\geq 4$.  Since $v$ is an internal vertex of $P_c$ and $c$ has no
neighbour on $P_b$, we conclude that there exist consecutive vertices $x,y$ on
$P_c$ such that $x\neq a$, $ux\in E(G)$, $uy\not\in E(G)$, and $y$ belongs to
the subpath of $P_c$ between $x$ and $c$ (possibly $y=c$.)  Note that $y$ has no
neighbours on $P_d$, since otherwise we contradict \ref{clm:5} for $P_c$. (For
this, recall that $d$ misses $P_c$, that $|P_c|\geq 4$, that $x\in P_c$ is
adjacent to $u\in P_b$, and that $ay\not\in E(G)$ because $P_c$ is an induced
path and $x$ belongs to the subpath of $P_c$ between $a$ and $y$.)

Let $P'_b$ denote the subpath of $P_b$ between $u$ and $b$. Let $P'_y=\{u,x,y\}$
and $P'_d=\{u\}\cup P_d$.  Note that $P'_b,P'_y,P'_d$ are paths of $G$, since
$u\not\in P_d$ by \ref{clm:3}.  Observe that $P'_y$ is missed by $d$, since
$P'_y\subseteq P_b\cup P_c$ and both $P_b$ and $P_c$ are missed by $d$.  Also,
$P'_y$ is missed by $b$, since $P_c$ is missed by $b$ and since $ub\not\in
E(G)$.  The fact that $ub\not\in E(G)$ also implies that $P'_d$ is missed by
$b$, since $P_d$ is.  Moreover, $P'_d$ is missed by $y$, since $uy\not\in E(G)$
and $y$ has no neighbour on $P_d$ as observed earlier.  Finally, note that
$y\not\in N(a)$ because $x\neq a$ and $P_c$ is an induced path.  Thus we
conclude that $y$ has no neighbour on $P_b$, since by \ref{clm:4} the only
possible neighbour of $y$ on $P_b$ is $u$, because $y\not\in N(a)$, but we have
$uy\not\in E(G)$ by the choice of $y$.  Thus $P'_b$ is missed by $y$, since
$P'_b$ is a subpath of $P_b$.  Moreover, $P'_b$ is missed by $d$, since $P_b$
is.  This shows that $u,b,y,d$ is a blocking quadruple. However, $|P'_b|+|P'_y|+
|P'_d|<|P_b|+|P_c|+|P_d|$, since $|P'_y|=3 < 4\leq |P_c|$ while $|P'_b|+
|P'_d|=|P_b|+|P_d|$. Thus we contradict the minimality of $a,b,c,d$.
\end{proofclaim}

A symmetric argument (using $P_b$, $P_d$ or $P_c$, $P_d$ in place of
$P_b$, $P_c$) yields:

\begin{claimplus}\label{clm:6}
Suppose that $|P_b|\geq 4$ and $|P_c|\geq 4$. Then there are no edges between
the internal vertices of $P_b$ and $P_c$.

Suppose that $|P_b|\geq 4$ and $|P_d|\geq 4$. Then there are no edges between
the internal vertices of $P_b$ and $P_d$.

Suppose that $|P_c|\geq 4$ and $|P_d|\geq 4$. Then there are no edges between
the internal vertices of $P_c$ and $P_d$.
\end{claimplus}

\begin{claim}
Let $v$ be the neighbour of $b$ on $P_b$. Then $v$ has a neighbour on $P_c$ or $P_d$.
\end{claim}

\begin{proofclaim}
For contradiction, suppose that no vertex of $P_c\cup P_d$ is adjacent to $v$.
Note that this implies $v\neq a$ and $va\not\in E(G)$.
Let $P_v$ denote the subpath of $P_b$ betwenn $a$ and $v$.  Clearly, $P_v$ is
missed by both $c$ and $d$, since $P_b$ is. Moreover, $P_c$ and $P_d$ are missed
by $v$, since no vertex of $P_c\cup P_d$ is adjacent to $v$.  Recall that $P_c$
is missed by $d$ and that $P_d$ is missed by $c$. Thus $a,v,c,d$ is a blocking
quadruple. However, since $b\not\in P_v$, we have $|P_v|<|P_b|$, and hence,
$|P_v|+|P_c|+|P_d|<|P_b|+|P_c|+|P_d|$, which contradicts the minimality of
$a,b,c,d$.

\end{proofclaim}

Again, by symmetry (between $b$, $c$, and $d$), we can conclude the following.

\begin{claimplus}\label{clm:9}
The neighbour of $b$ on $P_b$ has a neighbour on $P_c$ or $P_d$.

\quad\,The neighbour of $c$ on $P_c$  has a neighbour on $P_b$ or $P_d$.

\quad\,The neighbour of $d$ on $P_d$ has a neighbour on $P_b$ or $P_c$.
\end{claimplus}

The last claim in this series of statements is as follows.

\begin{claim}
If $|P_b|\geq 4$, then there are edges between the internal vertices
of $P_c$ and $P_d$.
\end{claim}

\begin{proofclaim}
Assume there are no edges between the internal vertices of $P_c$ and $P_d$. Let
$u$ be the neighbour of $b$ on $P_b$. Let $x$ be the neighbour of $a$ on $P_b$.
Clearly, $u\neq a$ and $ua,xb\not\in E(G)$, since $P_b$ is an induced path, and
$|P_b|\geq 4$.

By \ref{clm:9}, $u$ has a neighbour on $P_c$ or $P_d$. By symmetry, we may
assume, without loss of generality, that $u$ has a neighbour $v$ on $P_c$.  Note
that $v\neq a$, since $ua\not\in E(G)$. Also, $v\neq c$, since $c$ misses the
path $P_b$ and $u\in P_b$. Thus $v$ is an internal vertex of $P_c$.  From
\ref{clm:7}, we deduce that $vx\in E(G)$, since $vu\in E(G)$ and $x$ lies on the
subpath of $P_b$ between $a$ and $u$.  By \ref{clm:5}, there are
no edges between the internal vertices of $P_b$ and $P_d$, since $u\in P_b$ is
adjacent to $v\in P_c$ which are both internal vertices of the respective paths.
For the same reason, we have, by \ref{clm:6}, that $|P_c|=3$, namely, that
$P_c=\{a,v,c\}$,

Now, let $w$ denote the neighbour of $a$ on $P_d$.  Let $P_c'=\{x,v,c\}$, let
$P'_w=\{x,a,w\}$, and let $P_b'$ denote the subpath of $P_b$ between $x$ and
$b$.
Note that $vw\not\in E(G)$, since $v\in P_c$, $w\in P_d$, and we assume that
there are no edges between the internal vertices of $P_c$ and $P_d$.  So, we
conclude that $P_c'$ and $P_b'$ are missed by $w$, since $vw\not\in E(G)$ and
$c$ misses $P_d$, while $w\in P_d$, there are no edges between the internal
vertices of $P_b$ and $P_d$, and $b$ misses $P_d$.
Moreover, $P_b'$ is missed by $c$, since $P_b$ is.
Similarly, $P'_c$ and $P'_w$ are missed by $b$, since $P_c$ and $P_d$ are, and
since $xb\not\in E(G)$.  Likewise, $P'_w$ is missed by $c$, since $P'_w\subseteq
P_b\cup P_d$ and $c$ misses both $P_b$ and $P_d$.  This shows that $x,b,c,w$ is
a blocking quadruple, and we have $|P'_c|=|P_c|$, and $|P'_b|<|P_b|$ while
$|P'_w|\leq |P_d|$, since $w$ is an internal vertex of $P_d$. So $|P'_b| +
|P'_c| + |P'_w| < |P_b| + |P_c| + |P_d|$ contradicting our~choice~of~$a,b,c,d$.

\end{proofclaim}

As before, we deduce from this (by symmetry) the following.

\begin{claimplus}\label{clm:10}
If $|P_b|\geq 4$, there are edges between the internal vertices of $P_c$ and $P_d$.

~~~\,If $|P_c|\geq 4$, there are edges between the internal vertices of $P_b$ and $P_d$.

~~~\,If $|P_d|\geq 4$, there are edges between the internal vertices of $P_b$ and $P_c$.%
\end{claimplus}

We are finally ready to find an induced subgraph in $G$ isomorphic to one of the
configurations in Figure~\ref{fig:forb}.  To start, note that $a,b,c,d$ are
pairwise non-adjacent, since they form a blocking quadruple.  This implies that
the paths $P_b$, $P_c$, $P_d$ have each at least three vertices.

Let $u$ denote the neighbour of $b$ on $P_b$, let $v$ denote the neighbour of
$c$ on $P_c$, and let $w$ denote the neighbour of $d$ on $P_d$.  Clearly,
$a\not\in\{u,v,w\}$.

Suppose first that one of the three paths contains at least four vertices.
Without loss of generality (by the symmetry between the three paths), assume
that $|P_b|\geq 4$.  This implies that $ua\not\in E(G)$, since $P_b$ is an
induced path.  Further, by \ref{clm:9}, note that $u$ has a neighbour on
$P_c\cup P_d$.  Without loss of generality (again by symmetry), assume that $u$
has a neighbour on $P_c$. 

This implies, by \ref{clm:6}, that $|P_c|=3$. Namely, we have that
$P_c=\{a,v,c\}$.  Moreover, by \ref{clm:5}, there are no edges between the
internal vertices of $P_b$ and $P_d$, since $|P_b|\geq 4$  and $u\in P_b$ has a
neighbour on $P_c$.  In addition, recall that $c$ misses the path $P_b$, that
$u\in P_b$, and that $ua\not\in E(G)$.  Thus, since $u$ has a neighbour on
$P_c$, we conclude that $uv\in E(G)$. This implies, by \ref{clm:7}, that $v$ is
adjacent to all internal vertices on $P_b$.

Suppose that $vw\not\in E(G)$.  If $|P_d|\geq 4$, then $wc,wa\not\in E(G)$,
since $c$ misses $P_d$ and since $P_d$ is an induced path. Thus $w$ has no
neighbours on $P_c$, and also no neighbours on $P_b$, since $b$ misses $P_d$ and
there are no edges between the internal vertices of $P_b$ and $P_d$. This
contradicts \ref{clm:9}.  Therefore $|P_d|=3$ in which case $v$, $w$ are the
only internal vertices of $P_c$ and $P_d$, respectively. But we assume
$vw\not\in E(G)$, contradicting~\ref{clm:10}.

We must therefore conclude that $vw\in E(G)$. So, by \ref{clm:7}, the vertex $v$
is adjacent to all internal vertices of $P_d$.  Recall that $v$ is also adjacent
to all internal vertices of $P_b$, and there are no edges between the internal
vertices of $P_b$ and $P_d$.  Thus, since the paths $P_b$, $P_c$, $P_d$ are
induced, we conclude that the union $P_b\cup P_c\cup P_d$ of the three paths
induces in $G$ the graph d) in Figure~\ref{fig:forb}.  This completes the case
when one of the three paths has four or more vertices.

It remains to discuss the case when each of the three paths $P_b$, $P_c$, $P_d$
has exactly three vertices. Namely, we have $P_b=\{a,u,b\}$, $P_c=\{a,v,c\}$,
and $P_d=\{a,w,d\}$.  In this case, we show that their union $P_b\cup P_c\cup
P_d$ induces in $G$ one of the graphs d)-g) in Figure~\ref{fig:forb}. In
particular, if $uv,uw,vw\not\in E(G)$, then the paths induce the graph e) in
Figure~\ref{fig:forb}, while if $uv,uw,vw\in E(G)$, the paths induce the graph
g) in Figure~\ref{fig:forb}.  Similarly, if exactly one of $uv,uw,vw$ is in
$E(G)$, then the paths induce the graph f) in Figure~\ref{fig:forb}, while if
exactly two of $uv,uw,vw$ are in $E(G)$, the paths induce the graph d) in Figure
\ref{fig:forb} where the path labelled ``$\geq 2$ edges'' has exactly 2 edges. 
This exhausts all possibilities and concludes the proof of Lemma~\ref{lem:forb2}.
\end{proof}

Finally, to prove Theorem~\ref{thm:bq-nearly}, suppose that $G$ is a nearly
chordal graph. If $G$ contains a subgraph isomorphic to one of the graphs in
Figure~\ref{fig:forb}, then $G$ contains a blocking quadruple, as each of the
graphs in Figure~\ref{fig:forb} contains one on the vertices labelled as $a,b,c,d$.
Conversely, if $G$ does not contain as an induced subgraph any of the graphs in
Figure~\ref{fig:forb}, then $G$ is a nearly interval graph by
Lemma~\ref{lem:forb1}. So $G$ contains no blocking quadruple by
Lemma~\ref{lem:forb2}.  This proves the equivalence of (i) and (ii) of
Theorem~\ref{thm:bq-nearly}.

\section{Proof of Theorem~\ref{thm:main}}

Before the proof itself, we need to introduce some useful notions. We shall
describe a particular general construction of circular arcs, corresponding to
the vertices of a given chordal graph $G$, which will be based on a clique tree
$T$ of $G$ and its planar drawing. Intuitively, this operation will correspond
to cutting out the drawing of $T$ from the plane and stretching the resulting
hole to a disc. The cliques of $T$ will appear, possibly multiple times, in the
cyclic sequence of cliques on the boundary of this disc. We observe that the
same sequence of cliques can be obtained by a suitable depth-first traversal of
$T$ or an Euler tour of $T$ (when considered as a symmetric digraph).  For
convenience, we shall use the latter.

We then use this cyclic sequence to generate a family of circular-arcs.
The construction will produce circular arcs for all vertices of $G$, but we will
not be able to guarantee that the intersections of these arcs correspond to the
edges in $G$.  To ensure this, we introduce a set of conditions that will 
suffice to imply that the intersection graph of the arcs will indeed be $G$.

Finally, to prove Theorem~\ref{thm:main}, we will explain how to choose a clique
tree of $G$ and an appropriate planar drawing of it (i.e., an appropriate Euler
tour) so that these conditions are fulfilled.  

\subsection{Preliminaries}
Let $G$ be a chordal graph. A {\em clique} of $G$ is a set of pairwise adjacent
vertices. A {\em maximal clique} of $G$ is a clique that is not contained in a
larger clique of $G$.

A {\em clique tree} of $G$ is a tree $T$ whose vertices are the maximal cliques
of $G$ such that for all pairs of maximal cliques $C,C'$ of $G$, if $C''$ is a
maximal clique on the path of $T$ between $C$ and $C'$, then $C''\supseteq C\cap
C'$. Every chordal graph has a clique tree, and a graph is chordal if and only
if it has a clique tree \cite{gavril74}.  Note that a chordal graph can have
multiple clique trees.

We shall assume that $G$ contains {\em no universal vertex} (vertex adjacent
to all other vertices of $G$). This does not change any of the subsequent
arguments as a circular-arc graph remains circular-arc on addition of universal
vertices.

A {\em clique cover} of $G$ is a collection of maximal cliques of $G$ such that
every vertex in $G$ belongs to at least one of the cliques in the clique cover.

The following lemmas are simple consequences of the respective definitions. We
leave the details of their proofs to the reader.

\begin{lemma}\label{obs:leaf}
If $T$ is a clique tree of $G$, and $C$ is a leaf of $T$,  then there exists a
vertex $u\in V(G)$ such that $u$ is in $C$ and in no other clique in $V(T)$.
\end{lemma}

\begin{lemma}\label{obs:leaves}
If ${\cal Q}$ is a clique cover of $G$, and $T$ is a clique tree of $G$, then every
leaf of $T$ is in $\mathcal{Q}$.
\end{lemma}

\subsection{Euler tour}

Let $T$ be a clique tree of $G$.  Note that $T$ can be considered as a symmetric
directed graph (where each edge is replaced by a pair of directed arcs with
opposite directions). Then each vertex in $T$ has its in-degree equal to its
out-degree. Thus if considered as such, $T$ is a connected Eulerian digraph, and
hence, has an {\em Euler tour} (a walk visiting all edges of $T$).

We shall write ${\cal A}=A_0,A_1,\ldots,A_{k-1},A_0$ for an Euler tour of
$T$ visiting nodes $A_0,A_1,\ldots,A_{k-1},A_0$ in this order.  Note that nodes
of $T$ may appear multiple times in ${\cal A}$ as different $A_i$s.
Nonetheless, every leaf of $T$ appears exactly once in ${\cal A}$.  For
$i,j\in\{0,\ldots,k-1\}$, we write ${\cal A}[i,j]$ to denote the circular
subsequence $A_i, A_{i+1},\ldots, A_j$ (indices mod $k$) of ${\cal A}$.  In
other words,

\hfill${\cal A}[i,j]=\left\{\begin{array}{l@{\quad}l}
A_i,A_{i+1},\ldots,A_j & \mbox{if $i\leq j$}\\ A_i,\ldots,A_{k-1},A_0,\ldots,A_j
& \mbox{if $i> j$} \end{array}\right.$\hfill\null

\noindent We write ${\cal A}(i,j)$ for ${\cal A}[i,j]\setminus\{A_i,A_j\}$.

Our main tool in proving Theorem~\ref{thm:main} is the following lemma.

\begin{lemma}\label{lem:main}
Let $T$ be a clique tree of $G$, and let 
${\cal A}=A_0,A_1,\ldots,A_{k-1},A_0$ be an Euler tour of $T$.
If there exists a mapping $\phi:V(G)\rightarrow \{0,\ldots,k-1\}$ such that\smallskip
\begin{compactitem}[~~($\star$0)]
\item $u\in A_{\phi(u)}$ for each $u\in V(G)$,\setreflabel{($\star$0)}\label{enum:f0}
\end{compactitem}
\smallskip
and such that each $uv\in E(G)$ satisfies at least one of the following conditions:
\smallskip
\begin{compactenum}[~~($\star$1)]
\item $ux\in E(G)$ for every $x\in V(G)$ such that $A_{\phi(x)}\in {\cal A}\big(\phi(u),\phi(v)\big)$,\setreflabel{($\star$\theenumi)}\label{enum:f1}
\item $vx\in E(G)$ for every $x\in V(G)$ such that $A_{\phi(x)}\in {\cal A}\big(\phi(u),\phi(v)\big)$,\setreflabel{($\star$\theenumi)}\label{enum:f2}
\item $ux\in E(G)$ for every $x\in V(G)$ such that $A_{\phi(x)}\in {\cal A}\big(\phi(v),\phi(u)\big)$,\setreflabel{($\star$\theenumi)}\label{enum:f3}
\item $vx\in E(G)$ for every $x\in V(G)$ such that $A_{\phi(x)}\in {\cal A}\big(\phi(v),\phi(u)\big)$,\setreflabel{($\star$\theenumi)}\label{enum:f4}
\end{compactenum}
\smallskip
\noindent then $G$ is a circular-arc graph.
\end{lemma}

\begin{proof}
In the subsequent text, all the subscript arithmetic is considered reduced modulo $k$.

We start with the following useful observation.

\begin{claim}\label{clm:break}
Let $u\in V(G)$. If $u$ is in $A_i$ and $A_j$ where
${\cal A}(i,j)\neq\emptyset$, and $u$ does not belong to any clique in ${\cal A}(i,j)$,~then

\begin{compactenum}[\quad(i)]
\item $A_i=A_j$,\setreflabel{(\roman{enumi})}\label{subclm:ieqj}
\item the path in $T$ from any clique in ${\cal A}[j,i]$ to any clique in
${\cal A}(i,j)$ passes through $A_i$, and
\setreflabel{(\roman{enumi})}\label{subclm:path}
\item at least one of the cliques in ${\cal A}(i,j)$ is a leaf of $T$.
\setreflabel{(\roman{enumi})}\label{subclm:leaf}
\end{compactenum}
\end{claim}

\begin{proofclaim}
Suppose $A_i\neq A_j$. Let $e$ be the edge of $T$ between $A_i$ and
$A_{i+1}$.  Note that $A_{i+1}\neq A_j$ because ${\cal A}(i,j)\neq\emptyset$ and
because no clique in ${\cal A}(i,j)$ contains $u$ while $A_j$ does.
Further note that, since $u$ is also in $A_i$, it is not possible that the
edge $e$ is again traversed (in the reverse direction) somewhere between
$A_{i+1}$ and $A_j$. This implies that $A_{i+1}$ lies on the path of $T$
between $A_i$ and $A_j$. However, this is impossible as $T$ is a clique
tree, and both $A_i$ and $A_j$ contain $u$ while $A_{i+1}$ does not. This
proves \ref{subclm:ieqj}.

For \ref{subclm:path}, we assume \ref{subclm:ieqj}, i.e., that $A_i=A_j$, and
first we show that also $A_{i+1}=A_{j-1}$. Let $\cal M$ denote the set of all
the cliques that appear in ${\cal A}(i,j)$.  As ${\cal A}(i,j)$ is a walk that
goes through all the cliques in $\cal M$ and through no other clique, $\cal M$
induces a connected subgraph in $T$. Note that $A_i=A_j$ does not belong to
$\cal M$ as it contains $u$ while no clique in $\cal M$ does, by our assumption.
If $A_{i+1}\neq A_{j-1}$, then the path in $\cal M$ between $A_{i+1}$ and
$A_{j-1}$ together with the edges $A_iA_{i+1}$ and $A_iA_{j-1}$ form a cycle in
$T$ (recall that $A_i=A_j$ and $A_i\not\in {\cal M}$), which is a contradiction.
Therefore, $A_{i+1}=A_{j-1}$.  Now, let $e$ again be the edge of $T$ between
$A_i=A_j$ and $A_{i+1}=A_{j-1}$.  The removal of the edge $e$ from $T$ results
in two trees; let us denote them $T_1$ and $T_2$, and by symmetry, assume that
$A_i\in V(T_1)$ while $A_{i+1}\in V(T_2)$. As $e$ is traversed in ${\cal A}$
exactly twice, once from $A_i$ to $A_{i+1}$ and once from $A_{j-1}$ to $A_j$, we
conclude that the vertex set of $T_1$ consists precisely of the cliques
appearing in ${\cal A}[j,i]$ while the vertices of $T_2$ are the cliques
appearing in ${\cal A}(i,j)$.  Note that, as $T$ is a tree, any path of $T$
between a vertex in $T_1$ and a vertex in $T_2$ must pass through the edge $e$.
Therefore, so does any path of $T$ from a clique in ${\cal A}[j,i]$ to a clique
in ${\cal A}(i,j)$. Any such path contains $A_i$, which~proves~\ref{subclm:path}.

Finally, for \ref{subclm:leaf}, recall $T_2$ from the previous paragraph, and
note that the vertex set of $T_2$ consists of all cliques appearing in ${\cal
A}(i,j)$.  Thus, to prove \ref{subclm:leaf}, it suffices to show that some
vertex of $T_2$ is a leaf of $T$.  If $T_2$ contains only one vertex, then this
vertex is itself a leaf of $T$, and we are done. Otherwise, $T_2$ has at least
two vertices, and so it has at least two leaves, at most one of which is
incident to the edge $e$.  Consequently any other leaf of $T_2$ is also a leaf
of $T$, which implies \ref{subclm:leaf}.
\end{proofclaim}

Now, let $\phi:V(G)\rightarrow \{0,\ldots,k-1\}$ be a mapping satisfying the
conditions of the lemma. Define circular arcs $\{{\cal S}_u\}_{u\in V(G)}$
for the vertices of $G$ as follows. 
Let $\lambda_0,\lambda_1,\ldots,\lambda_{k-1}$ be $k$ distinct points on the
circle, arranged in this clockwise order. For each vertex $u\in V(G)$, define
indices $\ell_u$ and $r_u$ such that\smallskip

\begin{compactitem}[~~--~]
\item $A_{\ell_u}$ is the first clique that does not contain $u$ in\\
\null\hfill$A_{\phi(u)-1}$, $A_{\phi(u)-2}$, \ldots, $A_0$, $A_{k-1}$,
$A_{k-2}$, \ldots, $A_{\phi(u)+1}$\hfill\null
\item $A_{r_u}$ is the first clique that does not contain $u$ in\\
\null\hfill$A_{\phi(u)+1}$, $A_{\phi(u)+2}$, \ldots, $A_{k-1}$, $A_0$, $A_1$,
\ldots, $A_{\phi(u)-1}$\hfill\null
\end{compactitem}\smallskip
Then define $\mathcal{S}_u$ to be the clockwise circular arc from
$\lambda_{\ell_u+1}$ to $\lambda_{r_u-1}$.
\medskip

We show that the arcs $\{{\cal S}_u\}_{u\in V(G)}$ constitute a circular-arc
representation of $G$. Namely, we prove that for all vertices $u,v\in V(G)$, the
arc ${\cal S}_u$ intersects the arc ${\cal S}_v$ if and only if $uv\in E(G)$.
This will prove the lemma.
We shall need the following property which can be deduced directly from the construction.
\begin{claim}\label{clm:f2}
Each $u\in V(G)$ satisfies $\lambda_{\phi(u)}\in{\cal S}_u$, and $u\in A_i$ whenever
$\lambda_i\in{\cal S}_u$.
\end{claim}

\begin{claim}\label{clm:f3}
If ${\cal S}_u\cap {\cal S}_v\neq\emptyset$ then $uv\in E(G)$.
\end{claim}

\begin{proofclaim}
Since ${\cal S}_u\cap {\cal S}_v\neq\emptyset$, let $i$ be such that $\lambda_i\in
{\cal S}_u\cap {\cal S}_v$.  This implies that $u,v\in A_i$ by \ref{clm:f2}.
Therefore $uv\in E(G)$, since $A_i$ is a clique of $G$.
\end{proofclaim}

\begin{claim}
If $uv\in E(G)$, then ${\cal S}_u\cap {\cal S}_v\neq\emptyset$.
\end{claim}

\begin{proofclaim}
For contradiction, assume ${\cal S}_u\cap{\cal S}_v=\emptyset$.  If
$\phi(u)=\phi(v)$, then $\lambda_{\phi(u)}\in{\cal S}_u\cap{\cal S}_v$ by
\ref{clm:f2}, but this contradicts our assumption that ${\cal S}_u\cap {\cal
S}_v=\emptyset$.  Thus $\phi(u)\neq\phi(v)$.

By the assumptions of the lemma, one of the conditions
\ref{enum:f1}-\ref{enum:f4} is satisfied for the edge $uv$.  Notice that the
four conditions are the same up to reversal of ${\cal A}$ and exchanging $u$ and
$v$.  Indeed, it can be easily seen that reversing ${\cal A}$ will yield the
same collection of circular-arcs up to some rotational symmetries. Therefore, we
shall assume without loss of generality that \ref{enum:f1} holds.

We claim that the cliques containing $u$ form a consecutive sequence in
${\cal A}[\phi(u),\phi(v)]$. Suppose not. Let $A_i$ be the first clique in
${\cal A}[\phi(u),\phi(v)]$ such that $u\not\in A_i$ and let $A_j$ be the first
clique in ${\cal A}[i,\phi(v)]$ that~contains~$u$.  By
\ref{clm:break}\ref{subclm:leaf}, this means that there is some clique
$Q$ in ${\cal A}(i-1,j)$ that is a leaf of $T$ and does not contain~$u$.  As
$Q$ is a leaf of $T$, by Lemma~\ref{obs:leaf}, there is a vertex $x$ that is in
$Q$ but in no other clique in $V(T)$. But as $u\not\in Q$, we have $ux\not\in
E(G)$. Also, as $Q$ is a leaf of $T$, it occurs exactly once in ${\cal A}$ and
this occurrence is $A_{\phi(x)}$ by \ref{enum:f0}. So we have a vertex
$x$ such that $A_{\phi(x)}\in{\cal A}(\phi(u),\phi(v))$ and $ux\not\in E(G)$,
which contradicts \ref{enum:f1}. Therefore, the cliques containing $u$ indeed
form a consecutive sequence in ${\cal A}[\phi(u),\phi(v)]$. In other words,
for every $A_i \in{\cal A}[\phi(u),\phi(v)]$ such that $u\in A_i$, each clique in
${\cal A}[\phi(u),i]$ also contains $u$. 

By the construction of ${\cal S}_u$, this implies the following.

\begin{claimaux}{$+$}\label{clm:plus}
each $A_i \in{\cal A}[\phi(u),\phi(v)]$ such that $u\in A_i$ satisfies $\lambda_i\in {\cal S}_u$.
\end{claimaux}

\noindent Now, consider the path $P$ in $T$ from $A_{\phi(u)}$ to $A_{\phi(v)}$.
Note that $u\in A_{\phi(u)}$ and $v\in A_{\phi(v)}$ by \ref{enum:f0}.  Thus
since $uv\in E(G)$ and since $T$ is a clique tree, there is some clique $Q$ on
$P$ that contains both $u$ and $v$.  We claim that there is an occurrence of
$Q$ in ${\cal A}[\phi(u),\phi(v)]$. This is trivially true if $Q=A_{\phi(u)}$
or $Q=A_{\phi(v)}$. Otherwise, note that removing $Q$ from $T$ results in a
forest in which $A_{\phi(u)}$ and $A_{\phi(v)}$ are in different components.
Therefore, any walk in $T$ from $A_{\phi(u)}$ to $A_{\phi(v)}$ has to pass
through $Q$. This implies that there is at least one occurrence of $Q$ in
${\cal A}[\phi(u),\phi(v)]$, which is a walk from $A_{\phi(u)}$ to $A_{\phi(v)}$
in $T$.  Let $A_s$ be one of these occurrences. 

Note that $u,v\in A_s=Q$.
As $u\in A_s$ and since $A_s\in{\cal A}[\phi(u),\phi(v)]$, we conclude that
$\lambda_s\in{\cal S}_u$ by \ref{clm:plus}. This implies that $\lambda_s
\not\in{\cal S}_v$, since ${\cal S}_u\cap{\cal S}_v=\emptyset$. Therefore, there
exists some clique in ${\cal A}[s,\phi(v)]$ that does not contain $v$.  Let
$A_j$ be the last such clique in ${\cal A}[s,\phi(v)]$ and let $A_i$ be the last
clique in ${\cal A}[s,j]$ that contains $v$. Recall that $v\in A_{\phi(v)}$ by
\ref{enum:f0}. Thus the definition of $j$ implies that $v\in A_{j+1}$ and $v\in
A_i$ but $v$ is in no clique in ${\cal A}(i,j+1)$. By
\ref{clm:break}\ref{subclm:ieqj}, we now have $A_i=A_{j+1}$. Also by
\ref{clm:break}\ref{subclm:leaf}, there exists $Q\in {\cal A}(i,j+1)$ that is a
leaf of $T$. Using Lemma~\ref{obs:leaf}, let $x$ be a vertex that is present in
$Q$ and in no other clique in $V(T)$. Since $x$ appears in $A_{\phi(x)}$ by
\ref{enum:f0}, this implies that $A_{\phi(x)}=Q$. Thus, we conclude that
$A_{\phi(x)}\in{\cal A}(i,j+1)$, and hence, $A_{\phi(x)}\in{\cal
A}(\phi(u),\phi(v))$, which yields $ux\in E(G)$ by \ref{enum:f1}.  From this we
deduce that $u\in Q$, since $Q$ is the only clique that contains $x$. Now,
\ref{clm:break}\ref{subclm:path} tells us that the path in $T$ between any
clique in ${\cal A}[j+1,i]$ and any clique in ${\cal A}(i,j+1)$ passes through
the clique $A_i=A_{j+1}$. As $u\in Q$ and $u\in A_{\phi(u)}$ by \ref{enum:f0},
the vertex $u$ is present in every clique on the path in $T$ between
$A_{\phi(u)}$ and $Q$. This implies that $u\in A_{j+1}$. Hence,
$\lambda_{j+1}\in{\cal S}_u$ by \ref{clm:plus}. For this recall that
$A_{j+1}\in{\cal A}[\phi(u),\phi(v)]$. Now, since $A_j$ is the last clique in
${\cal A}[s,\phi(v)]$ that does not contain $v$, we also have
$\lambda_{j+1}\in{\cal S}_v$ by the construction of ${\cal S}_v$. Thus we obtain
${\cal S}_u\cap{\cal S}_v \neq\emptyset$, which is a contradiction.
\end{proofclaim}

This proves that $\{{\cal S}_u\}_{u\in V(G)}$ is indeed a circular-arc
representation of $G$. Thus $G$ is a circular-arc graph as claimed.
That concludes the proof.
\end{proof}

As an illustration of the construction provided by this lemma, let us point the
reader to Figure~\ref{fig:ex2}.  In part a), we see a chordal graph, its clique
tree, and an Euler tour of this clique tree (indicated by the labels 0, 1,
\ldots, 19). Part b) illustrates a choice of $\phi$ that satisfies the
conditions of the lemma, and also shows the resulting circular-arc
representation of the graph. Part c) shows a different mapping $\phi$ that fails
the conditions and similarly shows corresponding circular arcs that fail to
correctly represent the graph.

\begin{figure}[t!]
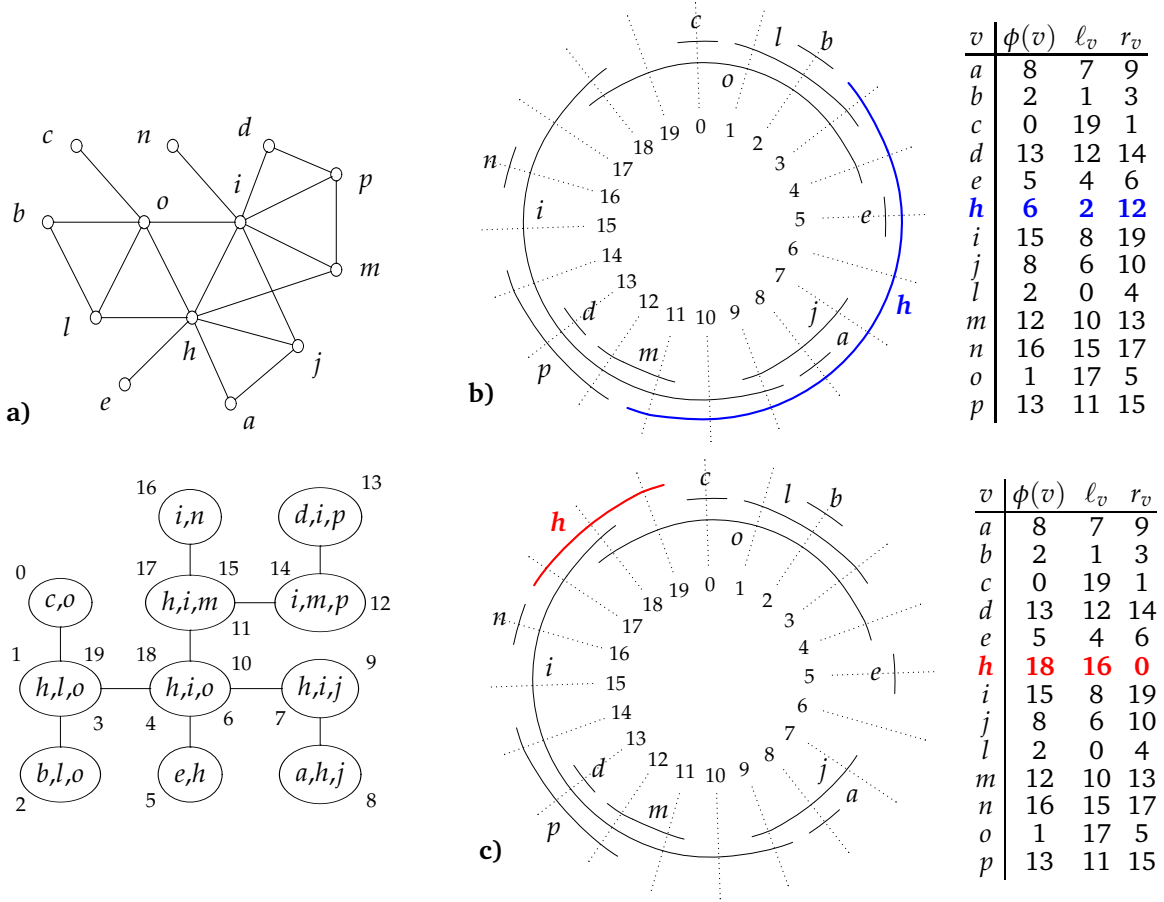

\centering
$\xy/r3pc/:(0,0)*[o][F]{\phantom{s}}="b";(1,0)*[o][F]{\phantom{s}}="o";
(2,0)*[o][F]{\phantom{s}}="i";(3,0.5)*[o][F]{\phantom{s}}="p";
(3,-0.5)*[o][F]{\phantom{s}}="m";(0.5,-1)*[o][F]{\phantom{s}}="l";
(1.5,-1)*[o][F]{\phantom{s}}="h";(0.8,-1.7)*[o][F]{\phantom{s}}="e";
(1.9,-1.9)*[o][F]{\phantom{s}}="a";(2.6,-1.3)*[o][F]{\phantom{s}}="j";
(0.3,0.8)*[o][F]{\phantom{s}}="c";(1.3,0.8)*[o][F]{\phantom{s}}="n";
(2.3,0.8)*[o][F]{\phantom{s}}="d";
{\ar@{-} "b";"o"};{\ar@{-} "b";"l"};{\ar@{-} "o";"l"};{\ar@{-} "o";"c"};
{\ar@{-} "o";"h"};{\ar@{-} "o";"i"};{\ar@{-} "i";"n"};{\ar@{-} "i";"d"};
{\ar@{-} "i";"p"};{\ar@{-} "i";"m"};{\ar@{-} "i";"h"};{\ar@{-} "i";"j"};
{\ar@{-} "m";"p"};{\ar@{-} "m";"h"};{\ar@{-} "p";"d"};{\ar@{-} "l";"h"};
{\ar@{-} "h";"e"};{\ar@{-} "h";"a"};{\ar@{-} "h";"j"};{\ar@{-} "a";"j"};
"a"+(0.2,-0.2)*{a};"b"+(-0.3,0.1)*{b};"c"+(-0.3,0.1)*{c};
"d"+(-0.25,0.2)*{d};"e"+(-0.2,-0.2)*{e};"h"+(-0.025,-0.35)*{h};
"i"+(-0.025,0.4)*{i};"j"+(0.2,-0.2)*{j};"l"+(-0.3,-0.1)*{l};
"m"+(0.35,0)*{m};"n"+(-0.3,0.1)*{n};"o"+(0.2,0.2)*{o};"p"+(0.3,-0.1)*{p};
(-0.3,-2)*{\mbox{\bf a)}};
\endxy$
\qquad\quad
{$\xy/r3pc/:
{{\xypolygon60"A"{~:{(1.9,0):}~>{}~={80}{}}},{\xypolygon60"B"{~:{(2.1,0):}~>{}~={80}{}}},
{\xypolygon60"C"{~:{(2.3,0):}~>{}~={80}{}}},{\xypolygon60"D"{~:{(1.5,0):}~>{}~={80}{}}},
{\xypolygon60"E"{~:{(1.7,0):}~>{}~={80}{}}},{\xypolygon60"F"{~:{(1.2,0):}~>{}~={80}{}}},
{\xypolygon60"Y"{~:{(1.0,0):}~>{}~={80}{}}}};
{"A40";"A38"**\crv{"A39"&}}; {"B58";"B56"**\crv{"B57"&}}; {"A4";"A2"**\crv{"A3"&}};
{"E25";"E23"**\crv{"E24"&}}; {"A49";"A47"**\crv{"A48"&}};
{"B55";"B29"**[blue][thicker]\crv{"B54"&"B51"&"B48"&"B45"&"B42"&"B39"&"B36"&"B33"&"B30"&}};
{"A37";"A8"**\crv{"A36"&"A33"&"A30"&"A27"&"A24"&"A21"&"A18"&"A15"&"A12"&"A9"}};
{"E43";"E35"**\crv{"E42"&"E39"&"E36"&}}; {"A1";"A53"**\crv{"A60"&"A57"&"A54"&}};
{"E31";"E26"**\crv{"E30"&"E27"&}}; {"B16";"B14"**\crv{"B15"&}};
{"E10";"E50"**\crv{"E9"&"E6"&"E3"&"E60"&"E57"&"E54"&"E51"&}};
{"B28";"B20"**\crv{"B27"&"B24"&"B21"&}};
"Y60"*{{_{\,1}}};"Y57"*{{_{\,2}}};"Y54"*{{_{\,3}}};"Y51"*{{_{\,4}}};"Y48"*{{_{\,5}}};"Y45"*{{_{6\,}}};
"Y42"*{{_{7\,}}};"Y39"*{{_{8\,}}};"Y36"*{{_{9}}};"Y33"*{{_{10}}};"Y30"*{{_{11}}};"Y27"*{{_{12}}};
"Y24"*{{_{13}}};"Y21"*{{_{14}}};"Y18"*{{_{15}}};"Y15"*{{_{16}}};"Y12"*{{_{17}}};"Y9"*{{_{18}}};
"Y6"*{{_{19}}};"Y3"*{{_{\,0}}};
"B3"*{c};"D1"*{o};"B59"*{l};"C57"*{b};"E48"*{e};"C44"*{\mbox{\color{blue}\bf\textit{h~}}};"A41"*{a};"D41"*{j};"D29"*{m};
"D24"*{d};"C25"*{p};"C15"*{n};"E17"*{i};
{\ar@{.} "F60";"C60"};{\ar@{.} "F57";"C57"};{\ar@{.} "F54";"C54"};{\ar@{.} "F51";"C51"};
{\ar@{.} "F48";"C48"};{\ar@{.} "F45";"C45"};{\ar@{.} "F42";"C42"};{\ar@{.} "F39";"C39"};
{\ar@{.} "F36";"C36"};{\ar@{.} "F33";"C33"};{\ar@{.} "F30";"C30"};{\ar@{.} "F27";"C27"};
{\ar@{.} "F24";"C24"};{\ar@{.} "F21";"C21"};{\ar@{.} "F18";"C18"};{\ar@{.} "F15";"C15"};
{\ar@{.} "F12";"C12"};{\ar@{.} "F9";"C9"};{\ar@{.} "F6";"C6"};{\ar@{.} "F3";"C3"};
(-1.3,-1.8)*{\mbox{\bf b)}};
\endxy$}
\quad{
\begin{tabular}{@{}c@{~{\vrule width 0.5pt}~}c@{~~}c@{~~}c@{}}
 $v$ & $\phi(v)$ & $\ell_v$ & $r_v$\\ \hline
$a$&8 & 7 & 9\\[-0.3ex] $b$&2 & 1 & 3\\[-0.3ex] $c$&0 & 19& 1\\[-0.3ex]
$d$&13 & 12 & 14\\[-0.3ex] $e$&5 & 4 & 6\\[-0.3ex]
\color{blue}\bf\textit{h}&\color{blue}\bf 6 &\color{blue}\bf 2 &\color{blue}\bf 12 \\[-0.3ex]
$i$&15 & 8 & 19\\[-0.3ex] $j$&8 & 6 & 10\\[-0.3ex] $l$&2 & 0 &4\\[-0.3ex]
$m$&12 & 10 & 13\\[-0.3ex] $n$&16 & 15 & 17\\[-0.3ex] $o$&1 & 17 & 5\\[-0.3ex]
$p$&13 & 11 & 15
\end{tabular}}
\medskip

\raisebox{6ex}{
$\xy/r2.7pc/:(0,0)*++[o][F]{c,\!o}="co";(1.5,0)*++[o][F]{h,\!i,\!m}="him";
(3,0)*++[o][F]{i,\!m,\!p}="imp";(0,-1)*++[o][F]{h,\!l,\!o}="hlo";
(1.5,-1)*++[o][F]{h,\!i,\!o}="hio";(3,-1)*++[o][F]{h,\!i,\!j}="hij";
(0,-2)*++[o][F]{b,\!l,\!o}="blo";(1.5,-2)*++[o][F]{e,\!h}="eh";
(3,-2)*++[o][F]{a,\!h,\!j}="ahj";(1.5,1)*++[o][F]{i,\!n}="in";
(3,1)*++[o][F]{d,\!i,\!p}="dip";
{\ar@{-} "co";"hlo"};{\ar@{-} "hlo";"blo"};{\ar@{-} "hlo";"hio"};
{\ar@{-} "hio";"eh"};{\ar@{-} "hio";"hij"};{\ar@{-} "hio";"him"};
{\ar@{-} "hij";"ahj"};{\ar@{-} "him";"imp"};{\ar@{-} "him";"in"};
{\ar@{-} "imp";"dip"};
"co"+(-0.45,0.35)*{_{\rm 0}};"hlo"+(-0.5,0.4)*{_{\rm 1}};
"blo"+(-0.45,-0.35)*{_{\rm 2}};"hlo"+(0.45,-0.4)*{_{\rm 3}};
"hio"+(-0.45,-0.4)*{_{\rm 4}};"eh"+(-0.45,-0.3)*{_{\rm 5}};
"hio"+(0.45,-0.35)*{_{\rm 6}};"hij"+(-0.45,-0.35)*{_{\rm 7}};
"ahj"+(0.6,-0.3)*{_{\rm 8}};"hij"+(0.6,0.3)*{_{\rm 9}};
"hio"+(0.6,0.3)*{_{\rm 10}};"him"+(0.6,-0.3)*{_{\rm 11}};
"imp"+(0.7,0)*{_{\rm 12}};"dip"+(0.6,0.4)*{_{\rm 13}};
"imp"+(-0.45,0.4)*{_{\rm 14}};"him"+(0.45,0.4)*{_{\rm 15}};
"in"+(-0.5,0.35)*{_{\rm 16}};"him"+(-0.5,0.4)*{_{\rm 17}};
"hio"+(-0.5,0.4)*{_{\rm 18}};"hlo"+(0.4,0.4)*{_{\rm 19}};
\endxy$}
\qquad\quad
{$\xy/r3pc/:
{{\xypolygon60"A"{~:{(1.9,0):}~>{}~={80}{}}},{\xypolygon60"B"{~:{(2.1,0):}~>{}~={80}{}}},
{\xypolygon60"C"{~:{(2.3,0):}~>{}~={80}{}}},{\xypolygon60"D"{~:{(1.5,0):}~>{}~={80}{}}},
{\xypolygon60"E"{~:{(1.7,0):}~>{}~={80}{}}},{\xypolygon60"F"{~:{(1.2,0):}~>{}~={80}{}}},
{\xypolygon60"Y"{~:{(1.0,0):}~>{}~={80}{}}}}; {"A40";"A38"**\crv{"A39"&}};
{"B58";"B56"**\crv{"B57"&}}; {"A4";"A2"**\crv{"A3"&}}; {"E25";"E23"**\crv{"E24"&}};
{"A49";"A47"**\crv{"A48"&}}; {"B13";"B5"**[red][thicker]\crv{"B12"&"B9"&"B6"&}};
{"A37";"A8"**\crv{"A36"&"A33"&"A30"&"A27"&"A24"&"A21"&"A18"&"A15"&"A12"&"A9"}};
{"E43";"E35"**\crv{"E42"&"E39"&"E36"&}}; {"A1";"A53"**\crv{"A60"&"A57"&"A54"&}};
{"E31";"E26"**\crv{"E30"&"E27"&}}; {"B16";"B14"**\crv{"B15"&}};
{"E10";"E50"**\crv{"E9"&"E6"&"E3"&"E60"&"E57"&"E54"&"E51"&}};
{"B28";"B20"**\crv{"B27"&"B24"&"B21"&}};
"Y60"*{{_{\,1}}};"Y57"*{{_{\,2}}};"Y54"*{{_{\,3}}};"Y51"*{{_{\,4}}};"Y48"*{{_{\,5}}};"Y45"*{{_{6\,}}};
"Y42"*{{_{7\,}}};"Y39"*{{_{8\,}}};"Y36"*{{_{9}}};"Y33"*{{_{10}}};"Y30"*{{_{11}}};"Y27"*{{_{12}}};
"Y24"*{{_{13}}};"Y21"*{{_{14}}};"Y18"*{{_{15}}};"Y15"*{{_{16}}};"Y12"*{{_{17}}};"Y9"*{{_{18}}};
"Y6"*{{_{19}}};"Y3"*{{_{\,0}}};
"B3"*{c};"D1"*{o};"B59"*{l};"C57"*{b};"E48"*{e};"C10"*{\mbox{\color{red}\bf\textit{h}}};"A41"*{a};"D41"*{j};"D29"*{m};
"D24"*{d};"C25"*{p};"C15"*{n};"E17"*{i};
{\ar@{.} "F60";"C60"};{\ar@{.} "F57";"C57"};{\ar@{.} "F54";"C54"};{\ar@{.} "F51";"C51"};
{\ar@{.} "F48";"C48"};{\ar@{.} "F45";"C45"};{\ar@{.} "F42";"C42"};{\ar@{.} "F39";"C39"};
{\ar@{.} "F36";"C36"};{\ar@{.} "F33";"C33"};{\ar@{.} "F30";"C30"};{\ar@{.} "F27";"C27"};
{\ar@{.} "F24";"C24"};{\ar@{.} "F21";"C21"};{\ar@{.} "F18";"C18"};{\ar@{.} "F15";"C15"};
{\ar@{.} "F12";"C12"};{\ar@{.} "F9";"C9"};{\ar@{.} "F6";"C6"};{\ar@{.} "F3";"C3"};
(-1.3,-1.8)*{\mbox{\bf c)}};
\endxy$}
\quad{
\begin{tabular}{@{}c@{~{\vrule width 0.5pt}~}c@{~~}c@{~~}c@{}}
 $v$ & $\phi(v)$ & $\ell_v$ & $r_v$\\ \hline
$a$&8 & 7 & 9\\[-0.3ex] $b$&2 & 1 & 3\\[-0.3ex] $c$&0 & 19& 1\\[-0.3ex]
$d$&13 & 12 & 14\\[-0.3ex] $e$&5 & 4 & 6\\[-0.3ex] {\color{red}
{\bf\textit{h}}}&{\color{red}\bf 18} & {\color{red}\bf 16} & {\color{red}\bf 0}
\\[-0.3ex]
$i$&15 & 8 & 19\\[-0.3ex] $j$&8 & 6 & 10\\[-0.3ex] $l$&2 & 0 &4\\[-0.3ex]
$m$&12 & 10 & 13\\[-0.3ex] $n$&16 & 15 & 17\\[-0.3ex] $o$&1 & 17 & 5\\[-0.3ex]
$p$&13 & 11 & 15
\end{tabular}}
\caption{{\em a)} Example graph, {\em b)} correct choice of $\phi$,
{\em c)} incorrect choice\label{fig:ex2}}
\end{figure}

\subsection{Proof}

Finally, we are ready to give a proof of Theorem~\ref{thm:main}.
Let $G$ be a chordal graph whose independence number $\alpha(G)$ is at most
four. The direction (i)$\Rightarrow$(ii) is proved as Lemma~\ref{lem:bqs}. 

For the converse, assume (ii), that is, $G$ contains no blocking quadruple.  In
what follows, we describe how to conclude that (i) holds, namely that $G$ is a
circular-arc graph. This splits into several cases.

First, if $\alpha(G)\leq 2$, then $G$ does not have an asteroidal triple. As $G$
is also chordal, $G$ is an interval graph~\cite{lekker62}, and hence, a
circular-arc graph.

Next, suppose that $\alpha(G)=3$. As $G$ is a perfect graph, there is a clique
cover ${\cal Q}=\{Q_1,Q_2,Q_3\}$ of $G$. Thus, by Lemma~\ref{obs:leaves}, any
clique tree $T$ of $G$ can have at most 3 leaves. If $T$ has only 2 leaves, then
$T$ is a path and therefore $G$ is an interval graph, and hence, a circular-arc
graph.
So, let us assume that $T$ has 3 leaves, which, in view of
Lemma~\ref{obs:leaves}, are $Q_1$, $Q_2$ and $Q_3$. Let ${\cal
A}=A_0,A_1,\ldots,A_{k-1},A_0$ be any Euler tour of $T$. It is easy to see that
each of $Q_1,Q_2,Q_3$, being a leaf of $T$, occurs exactly once in ${\cal A}$.
For every vertex $u\in V(G)$, define $Q(u)$ to be any clique among $Q_1,Q_2,Q_3$
that contains $u$ (break ties arbitrarily).  Then, define $\phi(u)$ to be the
integer in $\{0,1,\ldots,k-1\}$ such that $A_{\phi(u)}$ is the unique occurrence
of $Q(u)$ in ${\cal A}$.

Let us verify that the mapping $\phi$ satisfies the conditions of
Lemma~\ref{lem:main}.  Clearly, for each $u\in V(G)$, we have $u\in Q(u)$ by
definition, and hence, $u\in A_{\phi(u)}$ since $A_{\phi(u)}=Q(u)$.  Thus $\phi$
satisfies \ref{enum:f0}.  Now, consider an edge $uv\in E(G)$.  If
$\phi(u)=\phi(v)$, then ${\cal A}(\phi(u),\phi(v))={\cal A}(\phi(v),\phi(u))$ is
empty, and so each of \ref{enum:f1}-\ref{enum:f4} is vacuously satisfied.  So,
assume $\phi(u)\neq \phi(v)$ in which case either ${\cal A}(\phi(u),\phi(v))$ or
${\cal A}(\phi(v),\phi(u))$ contains no clique of ${\cal Q}$. For this, note
that $|{\cal Q}|=3$ and each clique in ${\cal Q}$ appears exactly once in $\cal
A$, while $A_{\phi(u)},A_{\phi(v)}\in {\cal Q}$ and $\phi(u)\neq \phi(v)$.
Thus, since each $x\in V(G)$ satisfies ${\cal A}_{\phi(x)}\in{\cal Q}$ by the
construction of $\phi$, we conclude that either \ref{enum:f1}-\ref{enum:f2} or
\ref{enum:f3}-\ref{enum:f4} are vacuously satisfied.  This verifies that $\phi$
satisfies the conditions of Lemma~\ref{lem:main}. Thus, by Lemma~\ref{lem:main},
we conclude that $G$ is a circular-arc graph.\medskip

We shall now prove the theorem for the case $\alpha(G)=4$.  As before, since $G$
is a perfect graph, this means that $G$ can be covered with 4 cliques. Let $T$
be a clique tree of $G$ and let $\mathcal{Q}= \{Q_1,Q_2,Q_3,Q_4\}$ be a clique
cover of $G$.  By Lemma~\ref{obs:leaves}, the tree $T$ has at most 4 leaves.  If
$T$ has only 2 leaves, then $T$ is a path and therefore, $G$ is an interval
graph (and hence a circular-arc graph).

\begin{figure}[t]
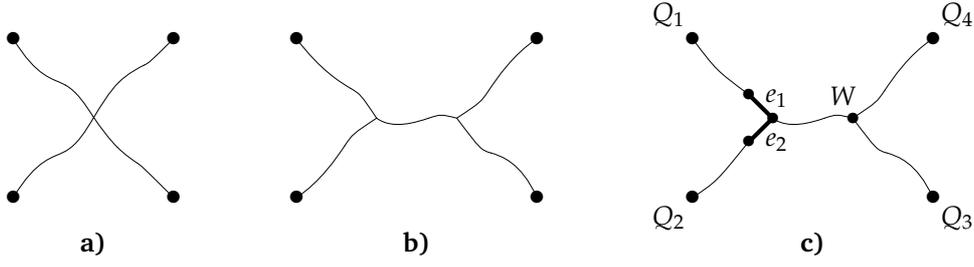

\centering
\begin{tabular}{c@{\qquad\qquad}c@{\qquad\qquad}c}
{$\xy/r2.5pc/: (0,0)*{}="x1";"x1"*{\mbox{\large\textbullet}}; (0,-2)*{}="x2";"x2"*{\mbox{\large\textbullet}};
(2,-2)*{}="x3";"x3"*{\mbox{\large\textbullet}}; (2,0)*{}="x4";"x4"*{\mbox{\large\textbullet}};
{"x1";"x3"**\crv{(0.3,-0.5)&(0.7,-0.5)&(1.0,-1.0)&(1.3,-1.4)&(1.6,-1.6)&}};
{"x2";"x4"**\crv{(0.3,-1.5)&(0.7,-1.5)&(1.0,-1.0)&(1.3,-0.5)&(1.7,-0.3)&}}; \endxy$}
&
{$\xy/r2.5pc/: (0,0)*{}="x1";"x1"*{\mbox{\large\textbullet}};
(0,-2)*{}="x2";"x2"*{\mbox{\large\textbullet}}; (3,-2)*{}="x3";"x3"*{\mbox{\large\textbullet}};
(3,0)*{}="x4";"x4"*{\mbox{\large\textbullet}}; {"x1";(1,-1)**\crv{(0.3,-0.4)&(0.8,-0.7)&}};
{"x2";(1,-1)**\crv{(0.3,-1.8)&(0.7,-1.2)&}}; {"x3";(2,-1)**\crv{(2.8,-1.55)&(2.3,-1.4)&}};
{"x4";(2,-1)**\crv{(2.7,-0.2)&(2.3,-0.8)&}}; {(1,-1);(2,-1)**\crv{(1.2,-1.15)&(1.8,-0.95)&}}; \endxy$}
&
{$\xy/r2.5pc/: (0,0)*{}="x1";"x1"*{\mbox{\large\textbullet}}; (0,-2)*{}="x2";"x2"*{\mbox{\large\textbullet}};
(3,-2)*{}="x3";"x3"*{\mbox{\large\textbullet}}; (3,0)*{}="x4";"x4"*{\mbox{\large\textbullet}};
(1,-1)*{}="Z";"Z"*{\bullet}; (2,-1)*{}="W";"W"*{\bullet}; (0.7,-0.7)*{}="e1";"e1"*{\bullet};
(0.7,-1.3)*{}="e2";"e2"*{\bullet}; {"x1";"e1"**\crv{(0.3,-0.4)&}}; {"x2";"e2"**\crv{(0.3,-1.8)&}};
{"x3";"W"**\crv{(2.8,-1.55)&(2.3,-1.4)&}}; {"x4";"W"**\crv{(2.7,-0.2)&(2.3,-0.8)&}};
{"Z";"W"**\crv{(1.2,-1.15)&(1.8,-0.95)&}}; "x1"+(-0.3,0.3)*{Q_1}; "x2"+(-0.3,-0.25)*{Q_2};
"x3"+(0.3,-0.25)*{Q_3}; "x4"+(0.3,0.3)*{Q_4}; "Z";"e1"**[thicker][thicker]\dir{-};
"Z";"e2"**[thicker][thicker]\dir{-}; "e1"+(0.35,-0.05)*{{e_1}}; "e2"+(0.35,0)*{{e_2}};
"W"+(-0.12,0.28)*{W}; \endxy$}
\\
{\bf a)} & {\bf b)} & {\bf c)}
\end{tabular}
\caption{When $T$ has 4 leaves\label{fig:tree4}}
\end{figure} 

Suppose that $T$ has 4 leaves. Then $Q_1,Q_2,Q_3,Q_4$ are the four leaves of $T$
and we have the two cases shown in Figure~\ref{fig:tree4} a), b).  For
$i\in\{1,2,3,4\}$, since $Q_i$ is a leaf of the clique tree $T$, we know from
Lemma~\ref{obs:leaf} that there exists a vertex $v_i\in Q_i$ such that $Q_i$ is
the only clique in $V(T)$ that contains $v_i$.

Now, consider the graph $H$ on $\{v_1,v_2,v_3,v_4\}$
with edge set\smallskip

\noindent\mbox{}\hfill$E(H)=\Big\{v_iv_j~\Big|~v_i,v_j$ avoid the vertices
$V(H)\setminus\{v_i,v_j\}\Big\}$\hfill\mbox{}\smallskip

\noindent As $v_1,v_2,v_3,v_4$ is not a blocking quadruple, $H$ is
either a 4-cycle, a $2K_2$ or has no edges.

For an Euler tour ${\cal A}$ of $T$, we shall write ${\cal A}|_{\cal Q}$ to
denote the cyclic subsequence of ${\cal A}$ induced by the cliques in ${\cal
Q}$. Namely, ${\cal A}|_{\cal Q}$ is the cyclic permutation of $Q_1,Q_2,Q_3,Q_4$
that we obtain from ${\cal A}$ by removing all cliques but the occurrences of
$Q_1,Q_2,Q_3,Q_4$.

Note that each $Q_i$, for $i\in\{1,2,3,4\}$, occurs exactly once in ${\cal A}$
as it is a leaf of $T$.  We say that an Euler tour ${\cal A}$ {\em respects} $H$
if $Q_i$ and $Q_j$ are cyclically consecutive in ${\cal A}|_{\cal Q}$ whenever
$v_iv_j\in E(H)$.  It turns out that the existence of such a tour implies
that $G$ is a circular-arc graph. This is proved as follows.

\begin{claimaux}{$*$}\label{clm:star}
If some Euler tour of $T$ respects $H$, then $G$ is a circular-arc graph.
\end{claimaux}

\noindent To see this, let ${\cal A}=A_0,A_1,\ldots,A_{k-1},A_0$ be such a tour.  Define $Q(u)$,
for every $u\in V(G)$, to be any clique among $Q_1,Q_2,Q_3,Q_4$ that contains
$u$ (break ties arbitrarily).  Then, define $\phi(u)$ to be the integer in
$\{0,1,\ldots,k-1\}$ such that $A_{\phi(u)}$ is the unique occurrence of $Q(u)$
in ${\cal A}$.

We claim that $\phi$ satisfies the conditions of Lemma~\ref{lem:main}.  Clearly,
for each $u\in V(G)$, we have $u\in{\cal A}_{\phi(u)}$, since $u\in Q(u)$ and
${\cal A}_{\phi(u)}=Q(u)$ by definition. Thus $\phi$ satisfies \ref{enum:f0}.
Now, for contradiction, suppose that there exists an edge $uv\in E(G)$ that does
not satisfy any of the conditions \ref{enum:f1}-\ref{enum:f4}.

Clearly, $Q(u) \neq Q(v)$ and $Q(u)$, $Q(v)$ do not appear consecutively in
${\cal A} |_{\cal Q}$, since otherwise either ${\cal A}(\phi(u),\phi(v))$ or
${\cal A}(\phi(v),\phi(u))$ contains no clique of $\cal Q$ and thus one of
\ref{enum:f1}-\ref{enum:f4} vacuously satisfied. So let us assume without loss
of generality that $Q(u)=Q_1$ and $Q(v)=Q_2$. Then, $\mathcal{A}|_{\cal Q}$ is
the cyclic ordering $Q_1,Q_3, Q_2,Q_4$ or its reverse. As none of the conditions
\ref{enum:f1}-\ref{enum:f4} of Lemma~\ref{lem:main} are satisfied for the edge
$uv$, it follows that there exists a vertex $x_3$ with $Q(x_3)=Q_3$ and a vertex
$x_4$ with $Q(x_4)=Q_4$ such that $ux_3\not\in E(G)$ and $ux_4\not\in E(G)$.
From this, we conclude that $u\not\in Q_3$ and $u\not\in Q_4$.  Similarly, there
exists $y_3$ with $Q(y_3)=Q_3$ and $y_4$ with $Q(y_4)=Q_4$ such that
$vy_3,vy_4\not\in E(G)$. Thus $v\not\in Q_3$ and $v\not\in Q_4$. From this,
since $Q_3$, resp. $Q_4$ is the only clique in $V(T)$ that contains $v_3$, resp.
$v_4$, we conclude that $uv_3,vv_3,uv_4,vv_4\not\in E(G)$. We now
have a path $v_1,u,v,v_2$ that is missed by $v_3$ and by $v_4$ (note that it is
possible that $u=v_1$ or $v=v_2$).  Therefore, $v_1,v_2$ avoid $v_3,v_4$, which
implies that $v_1v_2\in E(H)$. But $Q_1$ and $Q_2$ are not cyclically
consecutive in $\mathcal{A}|_{\cal Q}$, and so ${\cal A}$ does not respect $H$,
a contradiction.

This shows that $\phi$ satisfies the conditions of Lemma~\ref{lem:main}. Thus by
Lemma~\ref{lem:main}, we conclude that $G$ is a circular-arc graph.
This proves \ref{clm:star}.
\medskip

In light of \ref{clm:star}, it now suffices to show that an Euler tour of $T$
that respects $H$ exists. We analyze the two cases a) and b) of Figure
\ref{fig:tree4} separately.

Let us first consider case a). Since $T$ consists of
four paths joined at a single vertex, we have a freedom when traversing $T$ to
choose to follow the four paths in any order. In other words, for every possible
cyclic permutation of $Q_1,Q_2,Q_3,Q_4$, there is an Euler tour $\cal A$ of $T$
such that ${\cal A}|_{\cal Q}$ is precisely the chosen permutation.  From this,
it follows that there exists an Euler tour ${\cal A}$ of $T$ that respects $H$.
Namely, we choose a cyclic permutation of $Q_1,Q_2,Q_3,Q_4$ in which for
non-consecutive cliques $Q_i,Q_j$ the pair $v_iv_j$ is not an edge of $H$.  This
is always possible as $H$ is isomorphic to a $C_4$, $2K_2$ or $4K_1$. Therefore,
by \ref{clm:star}, we conclude that $G$ is a circular-arc graph.

If $G$ has a clique tree of the form a) from Figure~\ref{fig:tree4}, then we
are done. Hence, we shall assume that $G$ has no clique tree of the form a).
Suppose that $G$ has a clique tree of the form b) from Figure~\ref{fig:tree4}.
We claim that there is an Euler tour ${\cal A}$ of $T$ that respects $H$.  Let
the leaves of the tree $T$ be labelled as shown in c) in Figure~\ref{fig:tree4} and
let $W$ be the vertex of $T$ as shown in the figure.

It follows that there is no Euler tour of $T$ that respects $H$ if and only if
$H$ is the cycle $v_3,v_1,v_4,v_2$. Suppose that this is the case.  Note that
$v_1v_2\not\in E(H)$.  Let $e_1$ and $e_2$ be the edges shown in c) in
Figure~\ref{fig:tree4}, let $C$ be their common endpoint, and $C_1$ resp.  $C_2$
be their other endpoints, i.e.,  $e_1=CC_1$ and $e_2=CC_2$.  We claim that
either $C\cap C_1\subseteq W$ or $C\cap C_2\subseteq W$. Indeed, since $T$ is a
clique tree, we know that there exists a vertex $x_1\in C_1\setminus C$ and a
vertex $x_2 \in C_2\setminus C$. It follows that  $x_1\in Q_1$ and $x_2\in Q_2$,
since $T$ is a clique tree, $x_1,x_2\not\in C$, and since ${\cal Q}$ is a clique
cover of $G$. Now, if there exist vertices $u_1\in (C\cap C_1)\setminus W$ and
$u_2\in (C\cap C_2)\setminus W$, then $v_1,x_1,u_1,u_2,x_2,v_2$ is a path in $G$
from $v_1$ to $v_2$ that is missed by both $v_3$ and $v_4$. But then $v_1,v_2$
avoid $v_3,v_4$ implying $v_1v_2\in E(G)$, a contradiction. Thus, from $C\cap
C_1\subseteq W$ or $C\cap C_2\subseteq W$, we conclude that either the tree
obtained from $T$ by removing $e_1$ and adding the edge $C_1W$ or the tree
obtained from $T$ by removing $e_2$ and adding the edge $C_2W$ is another clique
tree of $G$.  However, both these trees are of the form~a), and we assume that
$G$ has no such clique tree, a contradiction.

We can therefore conclude that there is an Euler tour of $T$ that respects $H$.
Thus, by \ref{clm:star}, we again conclude that $G$ is a circular-arc
graph.\medskip

If $G$ has some clique tree that has four leaves, we are done.  Thus, we
shall assume that every clique tree of $G$ has exactly 3 leaves. Let $T$ be any
clique tree of $G$.  In this case, $T$ is of one of the two forms in
Figure~\ref{fig:tree3}.

\begin{figure}[t]
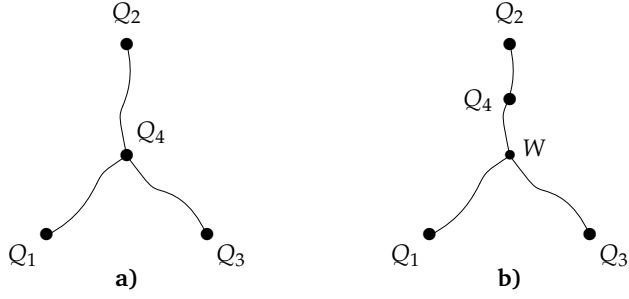

\centering\small
\begin{tabular}{c@{\qquad\qquad\qquad}c}
{$\xy/r2.5pc/: (1,0.4)*{}="x1";"x1"*{\mbox{\large\textbullet}};
(0,-2)*{}="x2";"x2"*{\mbox{\large\textbullet}}; (2,-2)*{}="x3";"x3"*{\mbox{\large\textbullet}};
(1,-1)*{}="x4";"x4"*{\mbox{\large\textbullet}}; {"x1";"x4"**\crv{(1.1,0)&(0.9,-0.5)&}};
{"x2";"x4"**\crv{(0.4,-1.8)&(0.7,-1.2)&}}; {"x3";"x4"**\crv{(1.8,-1.55)&(1.3,-1.4)&}};
"x1"+(0,0.4)*{Q_2}; "x2"+(-0.3,-0.25)*{Q_1}; "x3"+(0.3,-0.25)*{Q_3}; "x4"+(0.3,0.3)*{Q_4}; \endxy$}
&
{$\xy/r2.5pc/: (1,0.4)*{}="x1";"x1"*{\mbox{\large\textbullet}};
(0,-2)*{}="x2";"x2"*{\mbox{\large\textbullet}}; (2,-2)*{}="x3";"x3"*{\mbox{\large\textbullet}};
(1,-0.3)*{}="x4";"x4"*{\mbox{\large\textbullet}}; (1,-1)*{}="W";"W"*{\bullet};
{"x1";"W"**\crv{(1.1,0)&(0.9,-0.5)&}}; {"x2";"W"**\crv{(0.4,-1.8)&(0.7,-1.2)&}};
{"x3";"W"**\crv{(1.8,-1.55)&(1.3,-1.4)&}}; "x1"+(0,0.4)*{Q_2}; "x2"+(-0.3,-0.25)*{Q_1};
"x3"+(0.3,-0.25)*{Q_3}; "x4"+(-0.4,0)*{Q_4}; "W"+(0.3,0.1)*{W}; \endxy$}
\\
{\bf a)} & {\bf b)}
\end{tabular}
\caption{When $T$ has only three leaves\label{fig:tree3}}
\end{figure}

In both these cases, one of the cliques in ${\cal Q}=\{Q_1,Q_2,Q_3,Q_4\}$ is an
internal node of $T$. Let us assume without loss of generality that $Q_4$ is the
internal node among these in both cases. For $i\in\{1,2,3\}$, we let $v_i$
be a vertex that belongs to the clique $Q_i$ and to no other clique in $V(T)$.
We let $v_4$ be a vertex in $Q_4$ that does not belong to any other clique of
${\cal Q}$. Its existence is guaranteed by the fact that $G$ has no clique cover
of size 3 (as $\alpha(G)=4$).  For every vertex $u\in Q_1\cup Q_2\cup Q_3$, we
define $Q(u)$ so that $u\in Q(u)$. For a vertex $u\in V(G)$, we set $Q(u)=Q_4$ if
and only if $u\in Q_4\setminus (Q_1\cup Q_2\cup Q_3)$.

Let ${\cal A}=A_0,A_1,\ldots,A_{k-1},A_0$ be an Euler Tour of $T$.  We
again seek to construct $\phi$. For vertices $u$ with $Q(u)\in\{Q_1,Q_2,Q_3\}$,
the clique $Q(u)$ is a leaf of $T$ and thus appears exactly once in $\cal A$.
We choose this occurrence for $A_{\phi(u)}$ in order to satisfy \ref{enum:f0}.
For vertices $u$ with $Q(u)=Q_4$, we have more freedom as $Q_4$ appears multiple
times in any Euler tour of $T$. Namely, it appears exactly three times if $T$ is
of the form~a) of Figure~\ref{fig:tree3}, and it appears exacty two times if $T$
is of the form b). We choose one of these occurrences and assign $\phi(u)$ so
that $A_{\phi(u)}$ is this occurrence for every $u$ with $Q(u)=Q_4$. The choice
of this occurrence, however, will not be arbitrary. We will again use the graph
$H$ as constructed in the case with four~leaves~in~$T$.

Namely, let $H$ be the graph on vertices $\{v_1,v_2,v_3,v_4\}$ where $v_iv_j$ is
an edge if and only if $v_i,v_j$ avoid $V(H)\setminus\{v_i,v_j\}$.  Since
$v_1,v_2,v_3,v_4$ is not a blocking quadruple, we conclude that $H$ is either a
4-cycle, a $2K_2$, or edgeless.  Also, as in the case of four leaves, write
${\cal A}|_{\cal Q}$ to denote the cyclical sequence of the elements of $\cal
Q$ that is obtained from ${\cal A}$ by removing all occurrences of cliques that
are not in ${\cal Q}$.

First, suppose that $T$ is of the form a). Since $H$ contains no triangle, one
of the pairs $v_1v_2$, $v_2v_3$, $v_1v_3$ is not an edge of $H$. By symmetry,
without loss of generality, we may assume that $v_1v_2\not\in E(H)$.  Also,
without loss of generality, we may assume that ${\cal A}$ is an Euler tour of
$T$ such that\smallskip

\hfill${\cal A}|_{\cal Q}=Q_1,Q_4,Q_2,Q_4,Q_3,Q_4$\hfill\null\smallskip

\noindent (If not, we simply reverse and/or cyclically shift ${\cal A}$ to achieve this.)

For every vertex $u\in V(G)$ such that $Q(u)=Q_4$, define $\phi(u)$ to be the
integer such that $A_{\phi(u)}$ is the occurrence of $Q_4$ in ${\cal A}$ after
$Q_1$ and before $Q_2$. For every vertex $u\in V(G)$ such that $Q(u)=Q_i$, where
$i\in\{1,2,3\}$, define $\phi(u)$ to be the integer such that $A_{\phi(u)}$ is
the unique occurrence of $Q_i$ in ${\cal A}$. 

Let us now verify that $\phi$ satisfies the conditions of Lemma~\ref{lem:main}.
By definition, $u\in A_{\phi(u)}$ for all $u\in V(G)$. Thus $\phi$
satisfies \ref{enum:f0}. Now, consider an edge $uv\in E(G)$.
If $Q(u)=Q(v)$ or if $Q(u)\in\{Q_1,Q_2\}$ and $Q(v)\in\{Q_3,Q_4\}$, then either
${\cal A}(\phi(u),\phi(v))$ or ${\cal A}(\phi(v),\phi(u))$ contains only an
occurrence $A_i$ of $Q_4$ (if any) among the cliques of $\cal Q$, but no $x\in
V(G)$ satisfies $\phi(x)=i$ by the definition of $\phi$. Thus one of
\ref{enum:f1}-\ref{enum:f4} is vacuously satisfied in this case.  The same holds by symmetry
if $Q(u)\in\{Q_3,Q_4\}$ and $Q(v)\in\{Q_1,Q_2\}$.  Thus we may assume that
$Q(u)\neq Q(v)$ and either $Q(u),Q(v)\in\{Q_1,Q_2\}$ or
$Q(u),Q(v)\in\{Q_3,Q_4\}$.

Suppose first that $Q(u),Q(v)\in\{Q_3,Q_4\}$. By symmetry, we may assume that
$Q(u)=Q_3$ and $Q(v)=Q_4$. If $uv$ fails all of \ref{enum:f1}-\ref{enum:f4},
then there exist vertices $x_1,x_2,y_1,y_2$ where $Q(x_1)=Q(y_1)=Q_1$ and
$Q(x_2)=Q(y_2)=Q_2$ such that $ux_1,ux_2,vy_1,vy_2\not\in E(G)$.  From this, we
deduce $u,v\not\in Q_1$ and $u,v\not\in Q_2$. Therefore
$uv_1,uv_2,vv_1,vv_2\not\in E(G)$, since $Q_1$, resp.  $Q_2$ is the only clique
that contains $v_1$, resp. $v_2$.  Thus $v_3,u,v,v_4$ is a path (possibly
$u=v_3$ or $v=v_4$) missed by $v_1$ and $v_2$.  This yields $v_1v_2\in
E(H)$, a contradiction.

Next, suppose that $Q(u),Q(v)\in\{Q_1,Q_2\}$. By symmetry, assume that
$Q(u)=Q_1$ and $Q(v)=Q_2$. If $uv$ fails \ref{enum:f1}-\ref{enum:f4}, then there
are vertices $x_4,y_4$ with $Q(x_4)=Q(y_4)=Q_4$ such that $ux_4,vy_4\not\in
E(G)$. This implies $u,v\not\in Q_4$.  Let $P$ be the path of
$T$ between $Q_1$ and $Q_2$. Since $T$ is as depicted in Figure~\ref{fig:tree3},
the clique $Q_4$ lies on $P$.  Also, since $T$ is a clique tree and $uv\in E(G)$
where $u\in Q_1$ while $v\in Q_2$, there exists a clique $Q$ on $P$ such that
$u,v\in Q$. Thus either $u\in Q_4$ if $Q_4$ belongs to the subpath of $P$
between $Q_1$ and $Q$, or $v\in Q_4$ if otherwise. However, earlier we deduced
that $u,v\not\in Q_4$, a contradiction.

This shows that $\phi$ satisfies the conditions of Lemma~\ref{lem:main}.
Thus by Lemma~\ref{lem:main}, $G$ is a circular-arc graph.  \smallskip

Finally, suppose that $T$ is of the form b) of Figure
\ref{fig:tree3}.  Let the cliques $Q_1,Q_2,Q_3,Q_4$ be as labelled in the
figure. Let $P$ be the path between $Q_4$ and $Q_2$ in $T$.  Let $Q'_2$ be the
neighbour of $Q_4$ on $P$ (possibly $Q'_2=Q_2$).

We prove that the vertices $v_1,v_3$ avoid the vertices $v_2,v_4$ in $G$.
Firstly, suppose there exists a vertex $u\in Q'_2$ with $Q(u)=Q_4$.  Since $Q_4$
and $Q'_2$ are distinct maximal cliques of $G$, there exists $u'\in
Q'_2\setminus Q_4$. We claim that $Q(u')=Q_2$ and $u'\not\in Q_1\cup Q_3$.
Indeed, if $u'$ belongs to $Q_1$ or $Q_3$, then it also belongs to $Q_4$, since
$T$ is a clique tree, and the path between $Q_1,Q_3$ and $Q'_2$ goes through
$Q_4$. But $u'\notin Q_4$.  So we must conclude $u'\not\in Q_1\cup Q_3$, and
hence,  $Q(u')=Q_2$ because $Q_2$ is the only remaining clique from ${\cal Q}$
that can contain $u'$.  We also have $u\not\in Q_1\cup Q_3$, since
$Q(u)=Q_4$ if and only if $u\in Q_4\setminus(Q_1\cup Q_2\cup Q_3)$, by
definition. From this, we deduce that $uv_1,u'v_1,uv_3,u'v_3\not\in E(G)$,
since $Q_1$, resp.  $Q_3$ is the only clique in $V(T)$ that contains $v_1$,
resp. $v_3$.  This implies that $v_4,u,u',v_2$ forms a path in $G$ that is
missed by both $v_1$ and $v_3$  (note that possibly $u'=v_2$). Thus we conclude
that $v_1,v_3$ avoid $v_2,v_4$ as promised.

Secondly, suppose there exists a vertex $u\in Q_4\cap Q_2$ that does not belong to
$Q_1\cup Q_3$. Then $v_4,u,v_2$ forms a path in $G$ that is missed by both $v_1$
and $v_3$. Thus, we again conclude that $v_1,v_3$ avoid $v_2,v_3$.

We claim that one of these cases must be fulfilled.  Indeed, suppose that every
$u\in Q'_2$ satisfies $Q(u)\in\{Q_1,Q_2,Q_3\}$, and every $u\in Q_4\cap Q_2$
belongs to $Q_1\cup Q_3$.  Thus every $u\in Q_4\cap Q'_2$ is in $Q_1\cup Q_3$.
Namely, if $Q(u)\in\{Q_1,Q_3\}$, then $u\in Q_1\cup Q_3$ by definition, and if
$Q(u)=Q_2$, then $u\in Q_2$, and hence, $u\in Q_4\cap Q_2$ implying $u\in
Q_1\cup Q_3$ by our assumption. Since every such $u$ is in $Q_1\cup Q_3$, it
also belongs to $W$, since $T$ is a clique tree and $W$ lies on the path of $T$
between $Q_4$ and either of $Q_1$, $Q_3$.  This means that the tree obtained by
removing the edge $Q_4Q'_2$ from $T$ and adding the edge $WQ'_2$ is again a
clique tree of $G$. But this is a clique tree of $G$ with four leaves, and we
assume that $G$ has no such clique tree, a contradiction.

Now, since $v_1,v_3$ avoid $v_2,v_4$, we have that $v_1v_3$ and $v_2v_4$ are
edges of $H$.  Since $H$ contains no triangle, this means that one of $v_1v_2$, $v_2v_3$ is not
an edge of $H$. By symmetry, without loss of generality, assume that
$v_1v_2\not\in E(H)$. Also, without loss of generality, assume that ${\cal A}$
is an Euler tour of $T$ such that\smallskip

\hfill${\cal A}|_{\cal Q}=Q_1,Q_4,Q_2,Q_4,Q_3$\hfill\null\medskip

For every vertex $u\in V(G)$ such that $Q(u)=Q_4$, define $\phi(u)$ to be the
integer such that $A_{\phi(u)}$ is the occurrence of $Q_4$ in ${\cal A}$ after
$Q_1$ and before $Q_2$. For every vertex $u\in V(G)$ such that $Q(u)=Q_i$, where
$i\in\{1,2,3\}$, define $\phi(u)$ to be the integer such that $A_{\phi(u)}$ is
the unique occurrence of $Q_i$ in ${\cal A}$. 

Notice that this definition of $\phi$ is identical to that of the previous case
when $T$ was assumed to be type~a). We also assume that $v_1v_2\not\in E(H)$. 
This allows us to repeat the argument (word-for-word) from the previous case
that shows that $\phi$ satisfies the conditions of Lemma~\ref{lem:main}.  From
this, we conclude by Lemma~\ref{lem:main} that $G$ is a circular-arc graph.
This exhausts all cases and thus concludes the proof of Theorem~\ref{thm:main}.

\bibliographystyle{acm}
\bibliography{cca_arxiv}
\end{document}